\documentclass[11pt,reqno]{amsart}
\usepackage{amssymb}
\usepackage{amsmath, amssymb, amsthm}
\usepackage{mathrsfs}

\newtheorem{theorem}{Theorem}[section]
\newtheorem{definition}{Definition}[section]
\newtheorem{lemma}{Lemma}[section]
\newtheorem{remark}{Remark}[section]

\newtheorem{corollary}{Corollary}[section]
\numberwithin{equation}{section}

\newcommand{\supp}{{\mathrm {supp}}}

\begin{document}
\title[Euler-Boltzmann equations ]{On regular solutions of the $3$D compressible isentropic Euler-Boltzmann equations with vacuum }

\author{yachun li}
\address[Y. C. Li]{Department of Mathematics and Key lab of scientific and engineering computing (MOE), Shanghai Jiao Tong University,
Shanghai 200240, P.R.China} \email{\tt ycli@sjtu.edu.cn}

\author{Shengguo Zhu}
\address[S. G. Zhu]{Department of Mathematics, Shanghai Jiao Tong University,
Shanghai 200240, P.R.China;
School of Mathematics, Georgia Institute of Technology, Atlanta 30332, U.S.A.}
\email{\tt zhushengguo@sjtu.edu.cn}


\begin{abstract}
In this paper, we discuss the Cauchy problem for the compressible isentropic
Euler-Boltzmann equations with vacuum in radiation hydrodynamics. We establish the 
existence of unique local  regular solutions with vacuum by the theory of
quasi-linear symmetric hyperbolic systems and the Minkowski's inequality under some physical assumptions for the radition quantities. Moreover,  it is interesting to show the  non-global existence of regular solutions caused by the
effect of vacuum for polytropic gases with adiabatic exponent $1<\gamma\leq 3$ via some observation on the propagation of the radiation field.  Compared with \cite{tpy}\cite{tms1}\cite{xy}, some new initial conditions that will lead to the finite time blow-up for classical solutions have been introduced.
 These blow-up results tell us that the radiation effect on the fluid is not strong enough to  prevent the formation of singularities  caused by the appearance of vacuum.
\end{abstract}

\date{July 15, 2013}
\keywords{Euler-Boltzmann equations,  radiation, vacuum, symmetrization, classical solutions, regular solutions, local existence, blow up.\\
{\bf Acknowledgments.} The authors' research was supported in part
by Chinese National Natural Science Foundation under grants 11231006 and Natural Science Foundation of Shanghai under grant 14ZR1423100. Shengguo Zhu was also supported by China Scholarship Council.
}

\maketitle
\baselineskip=6.3mm

\section{Introduction}\label{S1}

The Euler-Boltzmann system appears in various astrophysical contexts \cite{kr} and in high-temperature plasma physics \cite{gp}. The couplings of fluid field and radiation field involve momentum source and energy source depending on the specific radiation intensity driven by the so called radiation transfer integro-differential equation \cite{gp}.
Suppose that the matter is in local thermodynamical equilibrium, the coupled system  of Euler-Boltzmann equations for the mass density $\rho(t,x)$, the fluid velocity $u(t,x)=(u_1,u_2,u_3)$, and the specific radiation intensity $I(v,\Omega,t,x)$ in three-dimensional  space reads as \cite{gp}
\begin{equation}
\label{eq:1.2}
\begin{cases}
\displaystyle
\frac{1}{c}\partial_tI+\Omega\cdot\nabla I=A_r,\\[10pt]
\displaystyle
\partial_{t}\rho+\nabla\cdot(\rho u)=0,\\[10pt]
\displaystyle
\partial_{t}\left(\rho u+\frac{1}{c^{2}}F_r\right)+\nabla\cdot(\rho u\otimes u+P_r)
  +\nabla p_m =0,
\end{cases}
\end{equation}
where $(t,x)\in \mathbb{R}^+\cap \mathbb{R}^3$, $v\in \mathbb{R}^+$ is the frequency of photons and $\Omega\in S^2$ is the travel direction of photons, here $S^{2}$ stands for the unit sphere in $\mathbb{R}^3$; $p_m$ is the material pressure satisfying the equation of state
\begin{equation}
\label{eq:1.3}
p_m=\rho^{\gamma}, \quad 1<\gamma \leq 3,
\end{equation}
where $\gamma$ is the adiabatic exponent.
\begin{equation}\label{right}
A_r=S-\sigma_aI
+\int_0^\infty \int_{S^{2}} \left(\frac{v}{v'}\sigma_s(v' \rightarrow v,\Omega'\cdot\Omega,\rho)I'
-\sigma_s(v \rightarrow v',\Omega\cdot\Omega',\rho)I\right) \text{d}\Omega' \text{d}v'
\end{equation}
is the collision term involving emission, absorption and scattering of energy, 
where  $I=I(v,\Omega,t,x)$, $I'=I(v',\Omega',t,x)$;  $S=S(v,t,x,\rho)$ is the rate of energy emission due to spontaneous process; $\sigma_a=\sigma_a(v,t,x,\rho)$ denotes the absorption coefficient that may  depend on the mass density $\rho$; $\sigma_s$ is the ``differential scattering coefficient'' (see \cite{lz} or  \cite{gp}) such that the probability of a photon being
scattered from $v'$ to $v$ contained in $\text{d}v$,  from $\Omega'$ to
$\Omega$ contained in $\text{d}\Omega$, and travelling a distance $\text{d}s$ is
given by $\sigma_s(v' \rightarrow v,\Omega'\cdot\Omega)\text{d}v \text{d}\Omega
\text{d}s$, and
\begin{equation*}
\begin{split}
\sigma_s\equiv  \sigma_s(v' \rightarrow v, \Omega'\cdot\Omega,\rho)=O(\rho),\ \sigma'_s\equiv  \sigma_s(v \rightarrow v', \Omega\cdot\Omega',\rho)=O(\rho).
\end{split}
\end{equation*}

For the isentropic flow, the impact of radiation on the dynamical properties of the fluid is described by the following two quantities
\begin{equation}
\label{eq:1.23}
\displaystyle
F_r=\int_0^\infty \int_{S^{2}} I(v,\Omega,t,x)\Omega \text{d}\Omega \text{d}v,\
P_r=\frac{1}{c} \int_0^\infty \int_{S^{2}}  I(v,\Omega,t,x)\Omega \otimes \Omega \text{d}\Omega \text{d}v,
\end{equation}
which are called the radiation flux and the radiation pressure tensor, respectively. The radiation field affects the dynamical properties of the fluid significantly, which makes it difficult to get the estimates of some physical quantities. For example, the material momentum $\int \rho u \text{d}x$ of the fluid is not conserved because of the impact coming from the radiation flux $F_r$ and the radiation pressure tensor $P_r$.

For pure compressible hydrodynamics equations without radiation, there have been many results on the local existence of regular solutions and the formation of singularities caused by the appearance of vacuum. The study on the appearance of vacuum in fluid dynamics can be traced back at least to the collected work of Von Neumann \cite{vn}, where  some important  remarks  on the general hydrodynamical discussion about motions  in one dimension following Riemann's theory have been maken.
Makino-Ukai-Kawashima \cite{tms1} discussed the Cauchy problem for the compressible Euler equations with both initial density and velocity compactly supported. They established the local existence of the regular solutions and showed that the life span is finite for any non-trivial solution, and the similar results have been obtained for  compressible Euler equations with damping in Liu-Yang \cite{tpy}.  Xu-Yang \cite{cy} established the local existence of smooth solutions to Euler equations with damping under the assumption of physical vacuum boundary condition.


Recently, similar problems for compressible radiation hydrodynamics equations started drawing attention of people. For Euler-Boltzmann equations, when the initial density is away from vacuum, Jiang-Zhong \cite{sjx} obtained the local existence of $C^1$ solutions for the Cauchy problem. Jiang-Wang \cite{pjd} showed that some $ C^1$ solutions will blow up in finite time regardless of the size of the initial perturbation. 

For Navier-Stokes-Boltzmann equations, Li-Zhu \cite{lz} considered the formation of singularities to classical solutions when the initial mass density is compactly supported, and  the existence of  unique local strong solutions for the isentropic flow to the Cauchy problem is proved in \cite{zl}.  Some special phenomenon has been observed. For example, it is known that  in contrast with the Second law of thermaldynamics, the associated entropy equation may contain a negative production term for RHD system, which has already been observed in Buet and Despr$\acute{\text{e}}$s \cite{add1}. Moreover, according to  Ducomet, Feireisl and Ne$\check{\text{c}}$asov$\acute{\text{a}}$ \cite{add2},  in which they  obtained the existence of global weak solution for some RHD model,  we know that the velocity field $u$ may develop uncontrolled time oscillations on the hypothetical vacuum zones where the mass density $\rho$ vanishes. Ducomet-Ne$\check{\text{c}}$asov$\acute{\text{a}}$ \cite{BD} \cite{BS} studied the global weak solutions and their large time behavior for one-dimensional case.

In this paper, we are interested in the isentropic Euler-Boltzmann equations with vacuum. Due to the complexity of this physical model, we first need to make some structure assumptions (see (\ref{xishu1})-(\ref{fg})) to the corresponding physical quantities such as $\sigma_a$, $\sigma_s$ and $S$, etc. Combining some arguments in Makino-Ukai-Kawashima \cite{tms1}, Jiang-Zhong \cite{sjx}, Xin-Yan \cite{xy} and introducing some new ideas dealing with the radiation field, we  prove the existence of unique local regular solutions to the Cauchy problem, and  we also study the formation
of singularities caused by  the occurrence of vacuum. 
  These blow-up results imply that the radiation effect on the fluid  is not strong enough to prevent the formation of singularities caused by the appearance of vacuum.  Some discussion on the relation between this kind of singularities and the formation of shock can be seen in \cite{tpy}.

 Compared with the results obtained in \cite{pjd} \cite{sjx} on the  formation of singularities away from vacuum for Euler-Boltzmann equations, via making full use of the mathematical structure of  our system  in vacuum domain, an important observation about the propagation of radiation effect has been shown in this paper as follows. If the initial velocity also vanishes in our vacuum region $V$  (see Definition \ref{local}) where the initial mass density vanishes, then there exists a critical time $T_c=\frac{2R_0}{c}$  such that beyond $T_c$,  the radiation effect on the fluid propagates completely from the interior of $B(t)$ (see Definition \ref{kobn}) to the exterior of $B(t)$. 
Based on this fact, we can remove the key assumptions in \cite{pjd}\cite{sjx} for the radiation field that 
\begin{equation*}\begin{cases}
I_0(v,\Omega,x) \equiv \overline{B}(v),\ \text{ for } \ x\cdot\Omega\leq0,\\[10pt]
 I_0(v,\Omega,x) \geq \overline{B}(v),\quad \forall (v,\Omega,x)\in \mathbb{R}^+\times S^{2}\times\mathbb{R}^3,
\end{cases}
\end{equation*}
where $\overline{B}(v)\in L^2(\mathbb{R}^+)$ is a function of $v$.

Another interesting  ingredients in this paper is our second type blow-up result (see Theorem \ref{th333} or Corollaries \ref{th3333}-\ref{th33333}) based on the definition of hyperbolic singularity set (see Definition \ref{bugers}), which is new even for the compressible Euler equations or Euler equations with damping. We  first recall the previous blow-up results obtained in \cite{tpy}\cite{tms1}\cite{xy} or Section $3.1$. They all need the assumption that 
\begin{equation}\label{haoaa}
u_0(x)=0,\quad \text{when} \quad \rho_0(x)=0\quad \text{in} \ \mathbb{E},
\end{equation}
where the domain $\mathbb{E}$ is an exterior domain  such as $B^C_{R_0}$ (\cite{tpy}\cite{tms1}) satisfying 
\begin{equation}\label{hao1aa}
\int_{B_{R_0}} \rho_0(x) \text{d}x >0;
\end{equation}
or an annular region (\cite{xy}, Section $3.1$) such as $B_{0}-A_{0}$ satisfying 
\begin{equation}\label{hao2aa}
\int_{A_{0}} \rho_0(x) \text{d}x >0.
\end{equation}
The assumption (\ref{haoaa}) will be used to control the size of  $\supp_x \rho(t,x)$ in \cite{tpy}\cite{tms1}, and  $B(t)$  in \cite{xy} or Section $3.1$. And  (\ref{hao1aa})-(\ref{hao2aa}) will be used to construct the desired ODE inequality such as (\ref{eq:3.12}) which plays a key role in the derivation of the contradiction.
So, from the blow-up results mentioned above,  the assumptions that $u_0=0$ in the vacuum domain and there must exist positive mass  (\ref{hao1aa})-(\ref{hao2aa}) contained in the intial  vacuum region seem to be necessary for the formation of singularities.  
However, our results obtained in Theorem \ref{th333}  or Corollaries \ref{th3333}-\ref{th33333}   answered this question negatively. This blow-up result is inspired by the defintion of regular solutions (see Definition \ref{666}) that:
$$
\partial_t u +u\cdot \nabla u=0,\quad \text{when} \ \rho=0,
$$
which tells us that in the vacuum domain,  the behavior of the velocity $u$ is actually controlled by a positive and symmetric hyperbolic system, i.e., the so called multi-dimensional Burgers equations. Then some assumption on  the spectrum of the Jacobian matrix of $u_0$ will lead to the finite blow-up, and the details can be seen in Section $3.2$. 

We organize this paper as follows. In section $2$, we first reformulate the
Cauchy problem for system (\ref{eq:1.2}) into a simpler form for the case $\sigma_s=0$, and then we establish the  existence of unique local regular solutions. In section $3$, we show that the regular solution obtained in Section $2$ will develop singularities in finite time provided that the initial data contain vacuum in some local domain for $\gamma > 1$. Finally, in Section $4$, we show the existence of unique local regular solutions for the case $\sigma_s\neq 0$ under some structure assumptions to the physical coefficients  $S$, $\sigma_a$ and $\sigma_s$.

\section{Reformulation and Local Existence}\label{S2}

\subsection{Reformulation}\ \label{S2.1}\\
We only consider the case $ \sigma_s=0 $ in this section. For the case $ \sigma_s\neq 0 $, some corresponding results will be shown in Section $4$. We reformulate the Cauchy problem of
the compressible Euler-Boltzmann equations (\ref{eq:1.2}) to a quasi-linear symmetric
hyperbolic system, so that we can get the local existence of regular solutions (see Definition \ref{666}).


From the assumptions of ``induced process'' and local thermal equilibrium (LTE, see  \cite{lz} \cite{gp}), $S$ and $\sigma_a$ can be written as
\begin{equation}\label{xishu1}
\begin{cases}
S(v,t,x,\rho)=K_a\overline{B}(v)\left(1+\frac{c^{2}I}{2hv^3}\right), \\[6pt]
\sigma_a(v,t,x,\rho)=K_a\cdot\left(1+\frac{c^{2}}{2hv^3}\overline{B}(v)\right),
\end{cases}
\end{equation}
where $\overline{B}(v)\in L^2(\mathbb{R}^+)$ is a function of $v$, $h$ is the Planck constant, and
\begin{equation} \label{fg}
K_a=K_a(v,t,x,\rho)=\rho \overline{}\overline{K}_a(v,t,x,\rho)=o(\rho)\geq 0,
\end{equation}
where $\overline{K}_a \in C^\infty$ for $(v,t,x,\rho)$ and $\lim_{\rho\rightarrow 0}\overline{K}_a(v,t,x,\rho)=0$.
More comments on $S(v,t,x,\rho)$ and $\sigma_a(v,t,x,\rho)$ can be seen in 
\cite{sjx} as well as in \cite{gp}. So, when $ \sigma_s=0 $,  the radiation transfer equation in (\ref{eq:1.2}) can be written as
\begin{equation} \label{eq:2.3}
\frac{1}{c}\partial_tI+\Omega\cdot\nabla I=-K_a\cdot (I-\overline{B}(v)).
\end{equation}

Then the compressible isentropic Euler-Boltzmann equations (\ref{eq:1.2}) can be reduced to
\begin{equation}\label{eq:2.4}
\begin{cases}
\displaystyle
\frac{1}{c}\partial_tI+\Omega\cdot\nabla I=-K_a\cdot (I-\overline{B}(v)),\\[8pt]
\displaystyle
\partial_{t}\rho+\nabla\cdot(\rho u)=0,\\[8pt]
\displaystyle
\partial_{t}(\rho u)+\nabla\cdot(\rho u\otimes u)+\nabla p_m=\frac{1}{c} \int_0^\infty \int_{S^{2}} K_a\cdot(I-\overline{B}(v))\Omega  \text{d}\Omega \text{d}v.
\end{cases}
\end{equation}

We consider the Cauchy problem with initial data
\begin{equation} \label{eq:2.211}
I|_{t=0}=I_0(v,\Omega,x), \  (\rho,u)|_{t=0}=(\rho_0(x),u_0(x)),
\end{equation}
where $\rho_0(x)\geq 0.$

We first introduce the definition of regular solutions to Cauchy problem (\ref{eq:2.4})-(\ref{eq:2.211}).
\begin{definition}\label{666} Let $T>0$ be a positive constant. A solution $(I,\rho,u)$ to Cauchy problem (\ref{eq:2.4})-(\ref{eq:2.211}) is called a regular solution if
\begin{equation*}\begin{split}
&(i)\,\,I(v,\Omega,t,x)\in  L^2\left(\mathbb{R}^+\times S^{2}; C^1( [0,T) \times \mathbb{R}^3)\right), \ (\rho(t,x),u(t,x))\in C^1\left([0,T)\times \mathbb{R}^3\right);\\
&(ii)\,\,\rho^{\frac{\gamma-1}{2}}\in C^1\left([0,T)\times \mathbb{R}^3\right), \  \rho \geq 0;\\
&(iii)\,\, \partial_t u+u\cdot \nabla u=0  \ \text{holds in the exterior of }\ \rm{supp}\rho.
\end{split}
\end{equation*}

\end{definition}
\begin{remark}\ \
\begin{enumerate}

  \item We point out that this definition for regular solutions is similar to  those of  Makio-Ukai-Kawashima \cite{tms1}   and Liu-Yang \cite{tpy} for $3$D compressible Euler equations, in which the corresponding existences of smooth solutions with vacuum have been proved. Moreover, $\sqrt{\gamma}\rho^{\frac{\gamma-1}{2}}$ is a very important physical quantity called local sound speed in gas dynamics.
  \item When $\rho> 0$, if $(\rho,u)$ has the regularity showed in $(\textrm{i})$-$(\textrm{ii})$, then it naturally holds that
\begin{equation}\label{jgh}\begin{split}
 \partial_t u+ u\cdot \nabla u+\frac{2\gamma}{\gamma-1}\rho^{\frac{\gamma-1}{2}} \nabla \rho^{\frac{\gamma-1}{2}} =\frac{1}{c} \int_0^\infty \int_{S^{2}} \overline{K}_a\cdot(I-\overline{B}(v))\Omega  \text{d}\Omega \text{d}v.
\end{split}
\end{equation}
Passing to the limit as $\rho\rightarrow 0$, we have
\begin{equation*}
 \partial_t u+ u\cdot \nabla u=\lim_{\rho\rightarrow 0}\Big(\frac{1}{c} \int_0^\infty \int_{S^{2}} \overline{K}_a\cdot(I-\overline{B}(v))\Omega  \text{d}\Omega \text{d}v\Big)=0.
\end{equation*}
So condition $(\textrm{iii})$ is reasonable at points $(t,x)(t>0)$ satisfying $\rho(t,x)=0$ due to the continuity of $\rho$ and properties of $\overline{K}_a$.
  \item We emphasize  condition $(\textrm{iii})$ is very important to make the velocity $u$ well defined at vacuum points and to ensure the uniqueness of regular solutions. Without condition $(\textrm{iii})$, it is very difficult to get enough information about velocity even when considering specific cases such as point vacuum or continuous vacuum of one piece.
\end{enumerate}
\end{remark}

Now we symmetrize hyperbolic system (\ref{eq:2.4}). After introducing a new variable
\begin{equation}\label{eq:2.5}
w=p^{\frac{\gamma-1}{2\gamma}}_m=\rho^{\frac{\gamma-1}{2}},
\end{equation}
system (\ref{eq:2.4}) can be written as
\begin{equation}
\begin{cases}
\label{eq:ccc}
\displaystyle
\frac{1}{c}\partial_tI+\Omega\cdot\nabla I=-K_a\cdot(I-\overline{B}(v)),\\[8pt]
\displaystyle
\partial_{t}w+u\cdot \nabla w+\frac{\gamma-1}{2}w \nabla \cdot u=0,\\[8pt]
\displaystyle
\frac{(\gamma-1)^2}{4\gamma}(\partial_tu+u\cdot\nabla u) +\frac{\gamma-1}{2}w\nabla w=\frac{(\gamma-1)^2}{4\gamma c}\int_0^\infty \int_{S^{2}} \overline{K}_a\cdot(I-\overline{B}(v))\Omega \text{d}\Omega \text{d}v.
 \end{cases}
\end{equation}
Let $U=U(t,x)=(w(t,x),u(t,x))$. Then system (\ref{eq:ccc}) is reduced to the following system of $(I,U)$:
\begin{equation}\label{eq:2.6}
\begin{cases}
\displaystyle
\frac{1}{c}\partial_tI+\Omega\cdot\nabla I=-K_a\cdot(I-\overline{B}(v)),\\[10pt]
 A_{0}(U)\partial_tU+\sum_{j=1}^{3} A_{j}(U)\partial_{x_j}U=G(I,U),
\end{cases}
\end{equation}
and the initial condition (\ref{eq:2.211}) turns into
\begin{equation} \label{eq:2.8}(I,U)|_{t=0}
        =(I_0(v,\Omega,x),w_0(x),u_0(x)),
\end{equation}
where $ w_0(x)=\rho_0(x)^{\frac{\gamma-1}{2}}\geq 0$, and
     $$ A_0(U)=\left(
         \begin{array}{cc}

            1& 0 \\[10pt]
           0& \frac{(\gamma-1)^2}{4\gamma}\mathbf{I}_3\\[10pt]
         \end{array}
       \right),\
 A_j(U)=\left(
         \begin{array}{cc}
         u_j  & \frac{\gamma-1}{2}w\delta_j \\[10pt]
         \frac{\gamma-1}{2}w\delta^\top_j & \frac{(\gamma-1)^2}{4\gamma}u_j \mathbf{I}_3 \\[10pt]
         \end{array}
       \right) \ (j=1,2,3), $$
where $\mathbf{I}_3$ is the $3\times 3$  unit matrix, $\delta_j=(\delta_{1j},\delta_{2j},\delta_{3j})$ is the Kronecker symbol satisfying $\delta_{ij}=1,\ i=j$ and $\delta_{ij}=0$, otherwise. The source terms $G=(G_0,G_1,G_2,G_3)$ are
\begin{equation*}
\begin{split}
&G_0(I,U)=0,\\
&G_j(I,U)=\frac{(\gamma-1)^2}{4c\gamma}\left(\int_0^\infty \int_{S^{2}} \overline{K}_a\cdot(I-\overline{B}(v))\Omega_j \text{d}\Omega \text{d}v\right)\ (j=1,2,3).
\end{split}
\end{equation*}
We note that $ A_j(U) ( j=0,1,2,3)$ are
$ C^\infty $ for $U$ and $G(I,U)$ are $C^\infty$ for $I$ and $U$. Moreover, $ A_j(U)\,\,(j=1,2,3)$ are all
symmetric, and $ A_0(U) $ is bounded and positively definite. In fact, we have
\begin{equation}\label{zhengding}
\frac{(\gamma-1)^2}{4\gamma}|\xi|^2 \leq (A_0(U)\xi, \xi) \leq |\xi|^2, \text{  for  }\xi\in \mathbb{R}^3.
\end{equation}
In order to get the local existence of regular solutions to the Cauchy problem (\ref{eq:2.4})-(\ref{eq:2.211}), it sufficies to prove the local existence of classcal solutions to the reformulated Cauchy problem (\ref{eq:2.6})-(\ref{eq:2.8}).

\subsection{Local existence and uniqueness of classical solutions to (\ref{eq:2.6})-(\ref{eq:2.8})}\ \label{S2.2}\\

We first introduce some notations. Denote by $W^{s,p}(\mathbb{R}^n)$ and $H^s(\mathbb{R}^n)$ the ordinary Sobolev spaces, and
$$
\|\cdot\|=\|\cdot\|_{L^2(\mathbb{R}^n)},\ \|\cdot\|_s=\|\cdot\|_{H^s(\mathbb{R}^n)},\ |\|\cdot\||_{s,T}=\max_{t\in[0,T]}\|\cdot\|_s.
$$

The following well-known estimates for the derivatives of product are useful in the energy estimates for local existence.
\begin{lemma} \label{pro:2.3}
Let
$
\frac{1}{r}=\frac{1}{p}+\frac{1}{q}, \quad 1\leq p,\ q,\ r \leq  +\infty
$.
For any integer $s \geq 0$, if
$
f \in W^{s,p}(\mathbb{R}^n),\  g\in W^{s,q}(\mathbb{R}^n)
$,
then
\begin{equation*} \|D^s(fg)\|_{L^r(\mathbb{R}^n)} \leq C_s\left(\|f\|_{L^p(\mathbb{R}^n)}\|D^sg\|_{L^q(\mathbb{R}^n)}+\|g\|_{L^p(\mathbb{R}^n)}\|D^sf\|_{L^q(\mathbb{R}^n)}\right),\end{equation*}
and when $s\geq 1$,
\begin{equation*}
\|D^s(fg)-fD^sg\|_{L^r(\mathbb{R}^n)}\leq C_s\left(\|Df\|_{L^p(\mathbb{R}^n)}\|D^{s-1}g\|_{L^q(\mathbb{R}^n)}+\|g\|_{L^p(\mathbb{R}^n)}\|D^{s}f\|_{L^q(\mathbb{R}^n)}\right),\end{equation*}
where $C_s$ is a constant depending on  $s$.
\end{lemma}

\begin{remark}
The proof of Lemma \ref{pro:2.3} can be found in \cite{amj}.
From Sobolev imbedding theorem we know that, if $s > \frac{n}{2}$, then we have
\begin{equation} \label{pro:2.300}
\|f\|_{L^\infty(\mathbb{R}^n)} \leq C_s \|f\|_{s},
\end{equation}
where $ f \in   L^\infty \bigcap H^s$. Then letting $r=2$ and $p=\infty,q=2$ in Lemma \ref{pro:2.3}, we obtain
\begin{equation}\label{pro:2.311}
\|D^\alpha(fg)\|\leq C_s \|f\|_{s}\|g\|_{s}\quad \text{for}\
s > \frac{n}{2},\  f,\  g \in H^s \  \text{and} \ |\alpha|\leq s.
\end{equation}

\end{remark}

To obtain the local existence of classical solutions to (\ref{eq:2.6})-(\ref{eq:2.8}), we need some assumptions on the coefficient $K_a$. If there exists a positive constant $M$ such that
$$ \|w^i\|_s\leq M,\quad \text{for}\quad i=1,2, $$ and we denote
$ K^i_a=K_a(v,t,x,w^i)$ and $\overline{K}^i_a=\overline{K}_a(v,t,x,w^i)\geq 0$, then we assume that
\begin{equation} \label{eq:2.10}
\begin{cases}
\displaystyle
\|K^i_a\|_{L^\infty(\mathbb{R}^+;C([0,T];H^s))}+\|\overline{K}^i_a\|_{L^2\cap L^\infty(\mathbb{R}^+;C([0,T];H^s))}\leq C_{s,M} \|w^i\|_s,\\[9pt]
\displaystyle
|\overline{K}_a(v,t,x,w^1)-\overline{K}_a(v,t,x,w^2)|\leq K(v,t,x)|w^1-w^2|,\\[9pt]
\displaystyle
\|K(v,t,x)\|_{L^\infty\cap L^2(\mathbb{R}^+;L^\infty([0,T]\times \mathbb{R}^3))}\leq C_{s,M}
\end{cases}
\end{equation}
for any $ t\in [0,T]$, where $M$ is a positive constant and $C_{s,M}$ is also a positive constant depending only on $s$ and $M$.

 Now we give the existence of unique local classical solutions to  (\ref{eq:2.6})-(\ref{eq:2.8}):
\begin{theorem}\label{th:1}
Let $s\geq 3$ be an integer.
If the initial data satisfy
\begin{equation*}\begin{split}
(I_0,U_0) \in \Phi:&=\left\{(I,U)|U(x)\in H^{s}(\mathbb{R}^3), \ (I(v,\Omega, x)-\overline{B}(v)) \in L^2\left(\mathbb{R}^+\times S^{2};H^{s}(\mathbb{R}^3)\right)\right\},
\end{split}
\end{equation*}
then there exists $T > 0$ such that Cauchy problem (\ref{eq:2.6})-(\ref{eq:2.8}) admits a unique classical solution $(I,U)$ satisfying
\begin{equation*}
\begin{split}
&U(t,x) \in C^1\left([0,T)\times \mathbb{R}^3\right) ),\\
&I(v,\Omega, t,x)-\overline{B}(v) \in L^2\left( \mathbb{R}^+\times S^{2}; C^1([0,T)\times \mathbb{R}^3)\right ).
\end{split}
\end{equation*}
\end{theorem}

\begin{remark}\label{ku}
The assumptions in Theorem \ref{th:1} for isentropic flows can be satisfied when the absorption coefficient is given by, for example (see \cite{sjx} \text{or}  \cite{gp}),
\begin{equation}\label{kk}
\begin{split}
&K_a(v,t,x,\rho)=D_1\rho \theta^{-\frac{1}{2}}\exp\left(-\frac{D_2}{\theta^{\frac{1}{2}}}\left(\frac{v-v_0}{v_0}\right)^2\right),\\
\end{split}
\end{equation}
where $\theta$ is the temperature, $v_0$ is the fixed frequency, $D_i(i=1,2)$ are positive constants. For isentropic polytropic gas, we know that $
p_m=R\rho\theta=\rho^\gamma$,
where $R$ is a positive constant. So we have $w=\rho^{\frac{\gamma-1}{2}}=\sqrt{R}\theta^{\frac{1}{2}}$, and
\begin{equation}
\begin{split}\label{kkc}
&\lim_{\rho\rightarrow 0}\frac{K_a(v,t,x,\rho)}{\rho}=\lim_{\theta \rightarrow 0} D_1\theta^{-\frac{1}{2}}\exp\left(-\frac{D_2}{\theta^{\frac{1}{2}}}\left(\frac{v-v_0}{v_0}\right)^2 \right)=0,\\[6pt]
&\lim_{\rho\rightarrow +\infty}\frac{K_a(v,t,x,\rho)}{\rho}=\lim_{\theta \rightarrow +\infty} D_1\theta^{-\frac{1}{2}}\exp\left(-\frac{D_2}{\theta^{\frac{1}{2}}}\left(\frac{v-v_0}{v_0}\right)^2\right)=0.
\end{split}
\end{equation}
Then in the case (\ref{kk}), we have
\begin{equation}\label{kkaa}
\begin{split}
&K_a(v,t,x,w)=D_1\sqrt{R}w^{\frac{3-\gamma}{\gamma-1}}\exp\left(-\frac{D_2\sqrt{R}}{w}\left(\frac{v-v_0}{v_0}\right)^2\right)=w^{\frac{2}{\gamma-1}}\overline{K}_a(v,t,\cdot,w),\\
&\overline{K}_a(v,t,x,w)=D_1\sqrt{R}\frac{1}{w}\exp\left(-\frac{D_2\sqrt{R}}{w}\left(\frac{v-v_0}{v_0}\right)^2\right).\\
\end{split}
\end{equation}
When $1< \gamma \leq 3$, it is easy to verify that the assumptions (\ref{eq:2.10}) in Theorem \ref{th:1} are satisfied if $s=3$.
\end{remark}

\textbf{Now we start proving Theorem \ref{th:1}.}
\begin{proof}
The proof is based on standard energy estimates as well as Banach contraction mapping principle.
Let $j(x)\in C^{\infty}_0(\mathbb{R}^3)$ be the standard mollifier satisfying
\begin{equation*}\begin{split}
&\supp j(x)\subseteq \{x:|x| \leq 1\},\ \int_{\mathbb{R}^3}j(x)\text{d}x=1, \ \forall \ j \geq 0.
\end{split}
\end{equation*}
Set $j_{\epsilon}=\epsilon^{-3}j(\frac{x}{\epsilon})$ and define $j_\epsilon u\in C^{\infty}$ by
\begin{equation*}
j_\epsilon u(x)=\int_{\mathbb{R}^3} j_\epsilon(x-y)u(y)dy.
\end{equation*}
For $k=0,1,2,...$, take $\epsilon_k=2^{-k}\epsilon_0 \ $ and
$$
U^{(k)}_0=j_{\epsilon_k}U_0(x), \ I^{(k)}_0=j_{\epsilon_k}I_0(v,\Omega,x),
$$
where $\epsilon_0$ is to be chosen later. We construct approximate solutions to (\ref{eq:2.6})-(\ref{eq:2.8}) through the following iteration scheme. We take
$$
U^{(0)}(t,x)=U^{(0)}_0(x), \ I^{(0)}(v,\Omega,t,x)=I^{(0)}_0(v,\Omega,x).
$$
For $k=0,1,...$, we define $U^{(k+1)}(t,x)$ and $I^{(k+1)}(v,\Omega,t,x)$ inductively as the solution of the following linearized problem:
\begin{equation}\label{eq:2.13}
\begin{cases}
\displaystyle
\frac{1}{c}\partial_t(I^{(k+1)}-\overline{B}(v))+\Omega\cdot\nabla (I^{(k+1)}-\overline{B}(v))=-K_a^{(k)}\cdot(I^{(k+1)}-\overline{B}(v)),\\[8pt]
\displaystyle
A_{0}(U^{(k)})\partial_tU^{(k+1)}+\sum_{j=1}^{3} A_{j}(U^{(k)}) \partial_{x_{j}}U^{(k+1)}=G(I^{(k)},U^{(k)}),\\[8pt]
I^{(k+1)}|_{t=0}=I^{(k+1)}_0(v,\Omega,x), \ U^{(k+1)}|_{t=0}=U^{(k+1)}_0(x),
\end{cases}
\end{equation}
where
\begin{equation}\label{jihao}
\begin{split}
&G_0=0,\ G_j(I^{(k)},U^{(k)})=\frac{(\gamma-1)^2}{4c\gamma}\left(\int_0^\infty \int_{S^{2}} \overline{K}^{(k)}_a\cdot(I^{(k)}-\overline{B}(v))\Omega_j \text{d}\Omega \text{d}v\right),\\
& \overline{K}^{(k)}_a=\overline{K}_a(v,t,x,w^{(k)}),\ K^{(k)}_a=K_a(v,t,x,w^{(k)}).
\end{split}
\end{equation}
%
It follows immediately that
\begin{equation}\label{eq:2.16}
U^{(k+1)}\in C^\infty\left([0,T_k]\times \mathbb{R}^3\right),\ I^{(k+1)}-\overline{B}(v)\in L^2\left((0,\infty)\times S^2;\ C^\infty([0,T_k]\times \mathbb{R}^3)\right)
\end{equation}
with $T_k$ being the largest time of existence for (\ref{eq:2.13}) where the estimates
\begin{equation}\label{eq:2.17}
 \int_{0}^{\infty}\int_{S^{2}} |\|I^{(k)}-I^{(0)}_0\||^2_{s,T_k} \text{d}\Omega \leq C \ \text{and}\ |\|U^{(k)}-U^{(0)}_0\||_{s,T_k}\leq C
\end{equation}
are valid for any given constant $C>0$. Hereinafter, $C$ stands for a generic positive constant.

Of course, in order to get the compactness, we have to guarantee that there exists a $ T>0 $ such that $ T_k \ge T $ for each $k$. So the following lemma which gives the uniform estimates of high order norms is very important.
\begin{lemma}[\textbf{Boundness of high order norms}]  \label{lemma:2:1}\ \\
There exist constants $C_1 > 0$ and $T_* > 0$ such that the solutions $(I^{(k)},U^{(k)})$ $(k=0,1,2,...)$ to (\ref{eq:2.13}) satisfy
\begin{equation}\label{eq:2.19}
 |\|U^{(k)}-U^{(0)}_0\||_{s,T_*}+|\|\partial_tU^{(k)}\||_{s-1,T_*}+\int_{0}^{\infty}\int_{S^{2}} |\|I^{(k)}-I^{(0)}_0\||^2_{s,T_*} \text{d}\Omega \text{d}v\leq C_1.
\end{equation}
\end{lemma}

\begin{proof} By induction, it is sufficient to prove that (\ref{eq:2.19}) holds for $(I^{(k+1)},U^{(k+1)})$ under the assumption that (\ref{eq:2.19}) holds for $(I^{(k)},U^{(k)})$. We divide the proof into three steps.

\underline{Step 1}. The estimate of $ \int_{0}^{\infty}\int_{S^{2}} |\|I^{(k+1)}-I^{(0)}_0\||^2_{s,T} \text{d}\Omega \text{d}v $.

Let $ Q^{(k+1)}=I^{(k+1)}-I^{(0)}_0 $. Then $ Q^{(k+1)}$ satisfies
\begin{equation} \label{eq:2.21}
\begin{cases}
\displaystyle
\frac{1}{c}\partial_tQ^{(k+1)}+\Omega\cdot\nabla Q^{(k+1)}
=-K_a^{(k)}Q^{(k+1)}+H^{(k)},\\
Q^{(k+1)}|_{t=0}=I^{(k+1)}_0-I^{(0)}_0 ,
\end{cases}
\end{equation}
where
$$
H^{(k)}=-K^{(k)}_a\cdot (I^{(0)}_0 -\overline{B}(v))-\Omega \cdot \nabla( I^{(0)}_0 -\overline{B}(v)).
$$
Differentiating the equation in (\ref{eq:2.21}) $\alpha$-times ($|\alpha|\leq s$) with respect to $x$ and multiplying the resulting equation by $D^{\alpha}Q^{(k+1)}$, we have
\begin{equation} \label{eq:2.23}
\ \frac{1}{2c}\frac{d}{\text{d}t}(D^{\alpha}Q^{(k+1)})^2+\frac{1}{2}\Omega\cdot\nabla (D^{\alpha}Q^{(k+1)})^2
=-D^{\alpha}(K^{(k)}_a Q^{(k+1)})D^{\alpha}Q^{(k+1)}
+D^{\alpha}H^{(k)}D^{\alpha}Q^{(k+1)}.
\end{equation}
Integrating (\ref{eq:2.23}) with respect to $x$ over $ \mathbb{R}^3$ and using the Cauchy's inequality, we have
\begin{equation} \label{eq:2.24}
\begin{split}
\frac{d}{\text{d}t}\int_{\mathbb{R}^3}|D^{\alpha}Q^{(k+1)}|^2dx
\leq&  C \|D^{\alpha}Q^{(k+1)}\|^2 +\|D^{\alpha}(K^{(k)}_aQ^{(k+1)})\|^2
+\|D^{\alpha}H^{(k)}\|^2 \\
=&:J_1+J_2+J_3.
\end{split}
\end{equation}
According to Lemma \ref{pro:2.3} and (\ref{pro:2.311}), we get
\begin{equation}\label{xxl}
J_2=\|D^{\alpha}(K^{(k)}_aQ^{(k+1)})\|^2
\leq C \|K^{(k)}_a\|^2_{s}\|Q^{(k+1)}\|^2_s,
\end{equation}
and
\begin{equation} \label{eq:2.26}
\begin{split}
J_3=&\|D^{\alpha}H^{(k)}\|^2
=\|D^{\alpha}(-K^{(k)}_a\cdot (I^{(0)}_0 -\overline{B}(v))-\Omega \cdot \nabla( I^{(0)}_0 -\overline{B}(v)))\|^2\\
\leq&  C \|K^{(k)}_a\|^2_{s}\|I^{(0)}_0-\overline{B}(v)\|^2_s +\| I^{(0)}_0-\overline{B}(v)\|_{s+1}^2.
\end{split}
\end{equation}
Combining (\ref{eq:2.10}) and (\ref{eq:2.24})-(\ref{eq:2.26}), we arrive at
\begin{equation*}\begin{split}
\frac{d}{\text{d}t}\|Q^{(k+1)}\|_s^2
\leq&  C(1+\|U^{(k)}\|_s)\|Q^{(k+1)}\|_s^2
+\|U^{(k)}\|_s\|I^{(0)}_0-\overline{B}(v)\|^2_s
+\|I^{(0)}_0-\overline{B}(v)\|^2_{s+1} \\
\leq & C\Big(\|Q^{(k+1)}\|_s^2+\|I^{(0)}_0-\overline{B}(v)\|^2_s\Big)+\|I^{(0)}_0-\overline{B}(v)\|^2_{s+1}.
\end{split}
\end{equation*}
By Gronwall's inequality we obtain
$$
|\|Q^{(k+1)}\||^2_{s,T}\leq e^{CT}\left(\|Q^{(k+1)}_0\|^2_s+T\big(\|I^{(0)}_0-\overline{B}(v)\|^2_s+\|I^{(0)}_0-\overline{B}(v)\|^2_{s+1}\big)\right
).$$
It is obvious that
$$
\int_{0}^{\infty}\int_{S^{2}}|\|Q^{(k+1)}\||^2_{s,T}\text{d}\Omega \text{d}v\leq e^{CT}\left(\int_{0}^{\infty}\int_{S^{2}}\|Q^{(k+1)}_0\|_s^2\text{d}\Omega \text{d}v+T\right).
$$
Taking $T=T_1$ to be small enough, we have
\begin{equation}\label{kka}
\int_{0}^{\infty}\int_{S^{2}} |\|I^{(k+1)}-I^{(0)}_0\||^2_{s,T_1} \text{d}\Omega \text{d}v\leq CC_1.
\end{equation}

\underline{Step 2}. The estimate of source terms $\|D^{\alpha}G(I^{(k)}, U^{(k)})\|$, $\forall|\alpha| \leq s $.

Due to Minkowski's inequality, Holder's inequality and (\ref{pro:2.311}), for $|\alpha| \leq s $, we have
\begin{equation}\label{eq:2.30}
\begin{split}
 \| D^{\alpha}G_j(I^{(k)},U^{(k)}) \|
  \leq &
C
 \Big\|D^\alpha \int_0^\infty \int_{S^{2}}    \overline{K}^{(k)}_a\cdot(I^{(k)}-\overline{B}(v))\Omega_j \text{d}\Omega \text{d}v \Big\|\\
\leq & C \int_0^\infty \int_{S^{2}}\|\overline{K}^{(k)}_a\|_s \|I^{(k)}-\overline{B}(v)\|_s \text{d}\Omega \text{d}v \\
\leq & C \left( \int_0^\infty \int_{S^{2}}\| \overline{K}^{(k)}_a\|^2_s \text{d}\Omega \text{d}v+\int_{0}^{\infty}\int_{S^{2}} \|I^{(k)}-\overline{B}(v)\|_s^2 \text{d}\Omega \text{d}v \right),
\end{split}
\end{equation}
which implies that
\begin{equation}\label{eq:2.33}
\begin{split}
\|G(I^{(k)},U^{(k)})\|_s \leq & C \left( \int_0^\infty \int_{S^{2}}\| \overline{K}^{(k)}_a\|^2_s \text{d}\Omega \text{d}v+\int_{0}^{\infty}\int_{S^{2}} \|I^{(k)}-\overline{B}(v)\|_s^2  \text{d}\Omega \text{d}v\right).
\end{split}
\end{equation}
Then from (\ref{kka}) and assumptions (\ref{eq:2.10}), we see that
\begin{equation}\label{eq:2.3rr}
\|G(I^{(k)},U^{(k)})\|_s \leq  C\left(\|U^{(k)}\|_s+\int_{0}^{\infty}\int_{S^{2}}\|I^{(k)}-\overline{B}(v)\|_s^2 \text{d}\Omega \text{d}v \right)\leq C_1.
\end{equation}
\underline{Step 3}. In order to estimate (\ref{eq:2.19}), define $M^{(k+1)}=U^{(k+1)}-U^{(0)}_0$, and it is easy to get
\begin{equation}\label{eq:2.35}
\begin{cases}
\displaystyle
 A_{0}(U^{(k)})\partial_tM^{(k+1)}+\sum_{j=1}^{3} A_{j}(U^{(k)})\partial_{x_{j}}{M^{(k+1)}}=G(I^{(k)},U^{(k)})+\overline{H}^{(k)},\\[10pt]
 M^{(k+1)}(x,0)=U^{(k+1)}_0(x)-U^{(0)}_0(x),
\end{cases}
\end{equation}
where
$$\overline{H}^{(k)}=-\sum_{j=1}^{3} A_{j}(U^{(k)})\partial_{x_{j}}{U^{(0)}_0}.$$
With the aid of the steps $1$ and $2$, it is easy to follow the standard procedure as in \cite{amj} and obtain that there exists a time $T_2$ such that
\begin{equation}\label{eq:2.36}
|\|U^{(k+1)}-U^{(0)}_0\||_{s,T_2}+|\|\partial_tU^{(k+1)}\||_{s-1,T_2}\leq C_1.
\end{equation}

Let $ T_*=\min\{T_1,T_2\} $ . Then the conclusions in Lemma \ref{lemma:2:1} are obtained.
\end{proof}

The following lemma implies that the operator associated with $(I^{(k)}, U^{(k)})$ is contracted.

\begin{lemma}\label{lemma:2:2}
There exist constants $T_{**} \in [0,T_*]$, $\eta < 1$ , $\{\beta_k\} \ (k=1,2,...)$  and  $\{\mu_k\} \ (k=1,2,...)$  with  $\sum_{k} |\beta_k|< +\infty$, and
$\sum_{k} |\mu_k|< +\infty $ , such that for each $k$
\begin{equation}\label{eq:2.37}
\begin{split}
&|\|U^{(k+1)}-U^{(k)}\||_{0,T_{\ast\ast}}+\left(\int_0^\infty \int_{S^{2}}|\|I^{(k+1)}-I^{(k)}\||^2_{0,T_{\ast\ast}} \text{d}\Omega \text{d}v\right)^{\frac{1}{2}}\\
\leq &\eta \left(|\|U^{(k)}-U^{(k-1)}\||_{0,T_{\ast\ast}}+\left(\int_0^\infty \int_{S^{2}}\||I^{(k)}-I^{(k-1)}\||^2_{0,T_{\ast\ast}} \text{d}\Omega \text{d}v\right)^{\frac{1}{2}}\right)+\beta_k+\mu_k.
\end{split}
\end{equation}
\end{lemma}

\begin{proof}
 According to the second equation in  (\ref{eq:2.13}), we have
 \begin{equation}\label{eq:2.38}
 \begin{split}
 \displaystyle
 &A_{0}(U^{(k)})\partial_t(U^{(k+1)}-U^{(k)})+\sum_{j=1}^3 A_{j}(U^{(k)})\partial_{x_{j}}(U^{(k+1)}-U^{(k)})\\
 =&G(I^{(k)},U^{(k)})-G(I^{(k-1)},U^{(k-1)})+Z^{(k)},
 \end{split}
\end{equation}
where
\begin{equation*}
 Z^{(k)}=-(A_{0}(U^{(k)})-A_{0}(U^{(k-1)}))\partial_t {U^{(k)}}-\sum{j=1}^3 (A_{j}(U^{(k)})-A_{j}(U^{(k-1)}))\partial_{x_{j}}{U^{(k)}}.
\end{equation*}

From Lemma \ref{lemma:2:1} and Taylor's expansion, we easily deduce that, $\forall \tau\in [0,T_*]$,
\begin{equation}\label{eq:2.39}
|\|Z^{(k)}\||_{0,\tau} \leq C|\|U^{(k)}-U^{(k-1)}\||_{0,\tau}.
\end{equation}
According to assumptions (\ref{eq:2.10}), and by Holder's inequality and Minkowski's inequality, we find that
%
\begin{equation}\label{eq:2.41}
\|G_0(I^{(k)},U^{(k)})-G_0(I^{(k-1)},U^{(k-1)})\| =0,
\end{equation}
and for $j=1,2,3$, $\forall \tau\in [0,T_*]$,
\begin{equation}\label{eq:2.41bn}
\begin{split}
&\|G_j(I^{(k)},U^{(k)})-G_j(I^{(k-1)},U^{(k-1)})(v,\Omega,\tau,\cdot)\|\\
\leq & C\int_0^\infty \int_{S^{2}}\|-(I^{(k-1)}-\overline{B}(v))\big(\overline{K}^{(k)}_a-\overline{K}^{(k-1)}_a\big)-\overline{K}^{(k)}_a\cdot(I^{(k)}-I^{(k-1)})\|\text{d}\Omega \text{d}v\\
\leq & C_{s,M} \|U^{(k)}-U^{(k-1)}\|\int_0^\infty \int_{S^{2}}\|I^{(k-1)}(v,\Omega,\tau,\cdot)-\overline{B}(v)\|_{s}\|K(v,\tau,\cdot)\|_{L^\infty(\mathbb{R}^3)}\text{d}\Omega \text{d}v\\
&+C\int_0^\infty \int_{S^{2}}\|\overline{K}_a(v,\tau,\cdot,w^{(k)})\|_{L^\infty(\mathbb{R}^3)}\|(I^{(k)}-I^{(k-1)})(v,\Omega,\tau,\cdot)\|\text{d}\Omega \text{d}v\\
\leq & C_{s,M}\|U^{(k)}-U^{(k-1)}\|\|K\|_{ L^2(\mathbb{R}^+;L^\infty([0,T]\times \mathbb{R}^3))}\left(\int_{0}^{\infty}\int_{S^{2}} \|I^{(k-1)}-\overline{B}(v)\|^2_s \text{d}\Omega \text{d}v\right)^{\frac{1}{2}}\\
&+C\|\overline{K}_a\|_{ L^2(\mathbb{R}^+;L^\infty([0,T]\times \mathbb{R}^3))}\left(\int_0^\infty \int_{S^{2}}\|I^{(k)}-I^{(k-1)}\|^2 \text{d}\Omega \text{d}v\right)^{\frac{1}{2}}\\
\leq & C\left(|\|U^{(k)}-U^{(k-1)}\||_{0,\tau}+\left(\int_0^\infty \int_{S^{2}}|\|I^{(k)}-I^{(k-1)}\||^2_{0,\tau} \text{d}\Omega \text{d}v\right)^{\frac{1}{2}}\right).
\end{split}
\end{equation}
Applying the standard energy estimates to (\ref{eq:2.38}) and using (\ref{eq:2.39})-(\ref{eq:2.41bn}), we easily get
\begin{equation}\label{eq:2.42}
\begin{split}
&|\|U^{(k+1)}-U^{(k)}\||_{0,\tau}\leq  e^{C\tau}\|U^{(k+1)}_0-U^{(k)}_0\|\\
&+e^{C\tau}\tau\left( |\|Z\||_{0,\tau}+|\|G(I^{(k)},U^{(k)})-G(U^{(k-1)},I^{(k-1)})\||_{0,\tau} \right).
\end{split}
\end{equation}
According to the properties of mollifier, for $\epsilon_0$ small enough, we know that
\begin{equation*}
\|J_\epsilon u-u\| \leq C \epsilon \|u\|_1, \quad \forall u \in H^1 ,\quad \epsilon \leq \epsilon_0.
\end{equation*}
So if we take $\epsilon_0<<1$,
then we have
\begin{equation}\label{eq:2.43}
\|U^{(k+1)}_0-U^{(k)}_0\|\leq C2^{-k}\|U_0\|_1.
\end{equation}
From (\ref{eq:2.42}), if we choose $T_3\in [0,T_*]$ to be small enough, then it is easy to get
\begin{equation}\label{eq:2.44}
\begin{split}
&|\|U^{(k+1)}-U^{(k)}\||_{0,T_3}\\
\leq & \eta_1 \left(|\|U^{(k)}-U^{(k-1)}\||_{0,T_3}+\left(\int_0^\infty \int_{S^{2}}\|I^{(k)}-I^{(k-1)}\|^2_{0,T_3} \text{d}\Omega \text{d}v\right)^{\frac{1}{2}}\right)+\beta_k,
\end{split}
\end{equation}
where $\eta_1 < \frac{1}{2}$ and $\sum_{k} |\beta_k|< \infty $.

To bound $I^{(k+1)}-I^{(k)}$, we use the first equation in (\ref{eq:2.13}) to see that
\begin{equation}\label{eq:2.45}
\begin{split}
&\frac{1}{c} \partial_t (I^{(k+1)}-I^{(k)}) +\Omega\cdot\nabla (I^{(k+1)}-I^{(k)})\\
=&(\overline{B}(v)-I^{(k)})(K^{(k)}_a-K^{(k-1)}_a)-K^{(k)}_a\cdot(I^{(k+1)}-I^{(k)}).
\end{split}
\end{equation}
It is easy to show that, $\forall \tau \in [0,T_*]$,
\begin{equation}\label{eq:2.455ttt}
\frac{d}{dt}\|I^{(k+1)}-I^{(k)}\|^2
\leq \|I^{(k)}-\overline{B}(v)\|_s\|K^{(k)}_a-K^{(k-1)}_a\| \|I^{(k+1)}-I^{(k)}\|,
\end{equation}
where we used the fact that $K_a\geq 0$.
From (\ref{eq:2.5}), Lemma \ref{lemma:2:1} and assumptions (\ref{eq:2.10}), we have
\begin{equation}\label{jia}
\begin{split}
&|K_a(v,t,x,w^{(k)})-K_a(v,t,x,w^{(k-1)})|\\
\leq& K(v,t,x)|w^{(k-1)}||w^{(k)}-w^{(k-1)}|+\overline{K}^{(k)}_a|w^{(k)}-w^{(k-1)}|\leq C|w^{(k)}-w^{(k-1)}|.
\end{split}
\end{equation}
Then using Young's inequality, we get
\begin{equation}\label{jiaa11}
\begin{split}
&\|I^{(k)}-\overline{B}(v)\|_s\|K^{(k)}_a-K^{(k-1)}_a\| \|I^{(k+1)}-I^{(k)}\|\\
\leq & C \|U^{(k)}-U^{(k-1)}\|^2\|I^{(k)}-\overline{B}(v)\|^2_s+\|I^{(k+1)}-I^{(k)}\|^2.
\end{split}
\end{equation}
Combining (\ref{eq:2.455ttt})-(\ref{jiaa11}), we have
\begin{equation}\label{eq:2.455}
\begin{split}
&\int_0^\infty \int_{S^{2}}|\|(I^{(k+1)}-I^{(k)})\||^2_{0,\tau} \text{d}\Omega \text{d}v\\
\leq&  e^{CT}\left(\int_0^\infty \int_{S^{2}}\|I^{(k+1)}_0-I^{(k)}_0\|^2 \text{d}\Omega \text{d}v
+C\tau|\|U^{(k)}-U^{(k-1)}\||^2_{0,\tau} \right).
\end{split}
\end{equation}
Similarly to the estimate of $\|U^{(k+1)}_0-U^{(k)}_0\|$, we easily get
\begin{equation*}
\left(\int_0^\infty \int_{S^{2}}\|I^{(k+1)}_0-I^{(k)}_0\|^2 \text{d}\Omega \text{d}v\right)^{\frac{1}{2}} \leq C 2^{-k}\left(\int_0^\infty \int_{S^{2}}\|I_0\|^2_1 \text{d}\Omega \text{d}v\right)^{\frac{1}{2}} .
\end{equation*}
If we choose $T_4\in [0,T_*]$ to be small enough, then we have
\begin{equation}\label{eq:2.46}
\begin{split}
&\left(\int_0^\infty \int_{S^{2}}|\|I^{(k+1)}-I^{(k)}\||^2_{0,T_4} \text{d}\Omega \text{d}v\right)^{\frac{1}{2}}\leq \eta_2|\|U^{(k)}-U^{(k-1)}\||_{0,T_4}
+\mu_k,
\end{split}
\end{equation}
where $\eta_2 < \frac{1}{2}$ and $\sum_{k} |\mu_k|< \infty $. Finally, taking $T_{\ast\ast}=\min\{T_3,T_4\}$, we obtain Lemma \ref{lemma:2:2} by adding (\ref{eq:2.44}) and (\ref{eq:2.46}) together.
\end{proof}

\textbf{Now we continue to prove Theorem \ref{th:1}.}

Lemma \ref{lemma:2:2} tells us that
\begin{equation*}\sum_{k=1}^{\infty}|\|U^{(k+1)}-U^{(k)}\||_{0,T_{\ast\ast}}+\Big(\int_0^\infty \int_{S^{2}}|\|I^{(k+1)}-I^{(k)}\||^2_{0,T_{\ast\ast}} \text{d}\Omega \text{d}v \Big)^{\frac{1}{2}} < +\infty, \end{equation*}
which implies that
\begin{equation}\begin{cases}\label{eq:2.47}
\displaystyle
\lim_{k\rightarrow\infty}|\|U^{(k+1)}-U^{(k)}\||_{0,T_{\ast\ast}}=0,\\
\displaystyle
\lim_{k\rightarrow\infty}\left(\int_0^\infty \int_{S^{2}}|\|I^{(k+1)}-I^{(k)}\||^2_{0,T_{\ast\ast}} \text{d}\Omega \text{d}v\right)^{\frac{1}{2}}=0.
\end{cases}
\end{equation}
In addition, from Lemma \ref{lemma:2:1} we know that sequence $ \{U^{(k)}(t,\cdot)\} \subset\subset \Phi $ for any fixed $t$ and
\begin{equation}\label{eq:2.49}
|\|U^{(k)}\||_{s,T_{\ast\ast}}+|\|\partial_{t} U^{(k)}\||_{s-1,T_{\ast\ast}} \leq 2C_1.
\end{equation}
Then from Sobolev interpolation inequalities, we have
\begin{equation}\label{eq:2.51}
\|U^{(k+1)}-U^{(k)}\|_{s'}\leq C \|U^{(k+1)}-U^{(k)}\|^{1-\frac{s'}{s}}\|U^{(k+1)}-U^{(k)}\|^{\frac{s'}{s}}_s.
\end{equation}
for any $ 0 < s'< s $.
So from (\ref{eq:2.49}) and (\ref{eq:2.51}), we get
\begin{equation*}
|\|U^{(k+1)}-U^{(k)}\||_{s',T_{\ast\ast}}\leq C |\|U^{(k+1)}-U^{(k)}\||^{1-\frac{s'}{s}}_{0,T_{\ast\ast}}, \quad \text{for any  } 0 < s'< s.
\end{equation*}
Similarly, using Sobolev interpolation inequalities and  Holder's inequality, we get
\begin{equation*}
\begin{split}
&\int_0^\infty \int_{S^{2}}|\|I^{(k+1)}-I^{(k)}\||^2_{s',T_{\ast\ast}} \text{d}\Omega \text{d}v\\
\leq& C  \int_0^\infty \int_{S^{2}}|\|I^{(k+1)}-I^{(k)}\||^{2(1-\frac{s'}{s})}_{0,T_{\ast\ast}} |\|I^{(k+1)}-I^{(k)}\||^{\frac{2s'}{s}}_{s,T_{\ast\ast}} \text{d}\Omega  \text{d}v\\
\leq &C  \left(\int_0^\infty \int_{S^2} |\|I^{(k+1)}-I^{(k)}\||^{2}_{0,T_{\ast\ast}}  \text{d}\Omega \text{d}v \right)^{\frac{s-s'}{s}} \left(\int_0^\infty \int_{S^{2}}|\|I^{(k+1)}-I^{(k)}\||^{2}_{s,T_{\ast\ast}}  \text{d}\Omega  \text{d}v \right)^{\frac{s'}{s}}\\
\leq &C\left(\int_0^\infty \int_{S^2} |\|I^{(k+1)}-I^{(k)}\||^{2}_{0,T_{\ast\ast}}  \text{d}\Omega \text{d}v \right)^{\frac{s-s'}{s}}
\end{split}
\end{equation*}
for any constant $s'$ satisfying  $ 0 < s'< s $.

According to  (\ref{eq:2.47}), we conclude that
\begin{equation}\label{eq:2.52}
\begin{cases}
\displaystyle
\lim_{k\rightarrow\infty}|\|U^{(k+1)}-U^{(k)}\||_{s',T_{\ast\ast}}=0,\\
\displaystyle
\lim_{k\rightarrow\infty}\left(\int_0^\infty \int_{S^{2}}|\|I^{(k+1)}-I^{(k)}\||^2_{s',T_{\ast\ast}} \text{d}\Omega \text{d}v\right)^{\frac{1}{2}}=0
\end{cases}
\end{equation}
for any $ 0 < s'< s $.

Therefore, if we choose $s' > \frac{5}{2} $, then from Sobolev embedding theorem, there exists $(I,U)$ such that
$$ U^{(k)} \rightarrow U \in C\left([0,T_{\ast\ast}];C^1(\mathbb{R}^3)\right),\ I^{(k)}\rightarrow I \in L^2\left(\mathbb{R}^+\times S^2; C([0,T_{\ast\ast}];C^1(\mathbb{R}^3))\right).$$
Furthermore, from the second equation in (\ref{eq:2.13}), we have
\begin{equation*}
\partial_t U^{(k+1)}=-A^{-1}_{0}(U^{(k)})\sum_{j=1}^3 A_{j}(U^{(k)})\partial_{x_j}U^{(k+1)}+A^{-1}_{0}(U^{(k)})G(I^{(k)}U^{(k)}),
\end{equation*}
then $\partial_t U^{(k+1)}\rightarrow \partial_t U \in C\left([0,T_{\ast\ast}]\times \mathbb{R}^3\right)$. Similarly, from the first equation in (\ref{eq:2.13}), we easily have $\partial_tI\in L^2\left(\mathbb{R}^+\times S^2; C([0,T_{\ast\ast}]\times \mathbb{R}^3)\right)$.
Thus, $(I,U)$ is a classical solution to (\ref{eq:2.6})-(\ref{eq:2.8}).

Finally, we consider the uniqueness of classical solutions. Let $(I,U)=(I,w,u)$ and $(\widehat{I},\widehat{U})=(\widehat{I},\widehat{w},\widehat{u})$ be two classical solutions to (\ref{eq:2.6})-(\ref{eq:2.8}). From the proof of Lemma \ref{lemma:2:2}, we have
 \begin{equation}\label{eq:2.38bnm}
 \begin{cases}
 \displaystyle
\frac{1}{c} \partial_t (I-\widehat{I}) +\Omega\cdot\nabla (I-\widehat{I})\\
\qquad=(\widehat{I}-\overline{B}(v))\big(K_a(v,t,x,\widehat{w})-K_a(v,t,x,w)\big)-K_a(v,t,x,w)\cdot(I-\widehat{I}).\\
\displaystyle
 A_{0}(U)\partial_t(U-\widehat{U})+\sum_{j=1}^{3} A_{j}(U)\partial_{x_{j}}(U-\widehat{U})
 =G(I,U)-G(\widehat{I},\widehat{U})+Z,
 \end{cases}
\end{equation}
where
\begin{equation*}
 Z=(A_{0}(\widehat{U})-A_{0}(U)\partial_t {\widehat{U}}+\sum_{j=1}^{3} (A_{j}(\widehat{U}-A_{j}(U)\partial_{x_{j}}{\widehat{U}}.
\end{equation*}
Similarly to the proof of Lemma \ref{lemma:2:2}, we can prove that $(I,U)=(\widehat{I},\widehat{U})$.

The proof of Theorem \ref{th:1} is finished.
\end{proof}

\begin{remark}\label{zhengxing}
by the standard method in Majda \cite{amj}, the classical solution obtained in the above theorem also satisfies
\begin{equation*}
\begin{split}
&U(t,x)=(w(t,x),u(t,x)) \in C^1\left([0,T); H^s (\mathbb{R}^3) \right)\cap C\left([0,T); H^{s-1} (\mathbb{R}^3) \right) ,\\
&I(v,\Omega, t,x)-\overline{B}(v) \in L^2\left( \mathbb{R}^+\times S^{2}; C^1([0,T); H^s (\mathbb{R}^3) )\cap C([0,T); H^{s-1} (\mathbb{R}^3) ) \right).
\end{split}
\end{equation*}
\end{remark}
\vspace{0.5cm}
Back to the Cauchy problem (\ref{eq:2.4})-(\ref{eq:2.211}), we will give the local existence and uniqueness of regular solutions based on the above results for classical solutions to Cauchy problem (\ref{eq:2.6})-(\ref{eq:2.8}).

\subsection{Local existence and uniqueness of regular solutions to (\ref{eq:2.4})-(\ref{eq:2.211})}\ \label{S2.3}\\
In this section, we will give the local existence  and uniqueness of regular solutions to the original Cauchy problem (\ref{eq:2.4})-(\ref{eq:2.211}) based on the results obtained in Section \ref{S2.2}.
\begin{theorem}\label{th:2.2}\
Let $s\geq 3$  be an integer.
If the initial data satisfy
\begin{equation*}\begin{split}
(I_0,\rho_0,u_0) \in \Psi:=&\big\{(I,\rho,u)|\ \rho(x) \geq 0 ;  \  \ (\rho^{\frac{\gamma-1}{2}},u)(x)\in H^{s}(\mathbb{R}^3),\\
& I(v,\Omega,x)-\overline{B}(v)\in L^2(\mathbb{R}^+\times S^{2};H^{s}(\mathbb{R}^3))\big\},
\end{split}
\end{equation*}
then there exists a time $T > 0$  such that Cauchy problem (\ref{eq:2.4})-(\ref{eq:2.211}) admits a unique regular solution $(I,\rho,u)$. 

\end{theorem}

\begin{proof}
From Theorem \ref{th:1}, we know that
there exists $T > 0$ such that Cauchy problem (\ref{eq:2.6})-(\ref{eq:2.8}) has a unique classical solution $(I,U)$ satisfying
\begin{equation}\label{reg2}
\begin{split}
U=(w,u)\in C^1\left([0,T)\times \mathbb{R}^3\right),\ I\in L^2\left(\mathbb{R}^+\times S^2; C^1([0,T)\times \mathbb{R}^3)\right).
\end{split}
\end{equation}
According to  transformation (\ref{eq:2.5}), since
$ \rho(t,x)=w^{\frac{2}{\gamma-1}} $ and
$ \frac{2}{\gamma-1}\geq 1$ due to $ 1< \gamma \leq3 $ , it is easy to show that
$(\rho,u)(t,x)\in C^1\left([0,T)\times\mathbb{R}^3\right)$.

Multiplying $(\ref{eq:ccc})_2$ by
$
\frac{\partial \rho}{\partial{w}}=\frac{2}{\gamma-1}w^{\frac{3-\gamma}{\gamma-1}}\in C\left([0,T)\times \mathbb{R}^3\right)
$,
we get
\begin{equation} \label{eq:2.58}
\partial_t\rho+u \cdot\nabla \rho+\rho\nabla\cdot u=0,
\end{equation}
which is exactly the continuity equation in (\ref{eq:2.4}). Multiplying $(\ref{eq:ccc})_3$ by
$\frac{4\gamma}{(\gamma-1)^2}w^{\frac{2}{\gamma-1}}\in C^1\left([0,T)\times \mathbb{R}^3\right)$,
we get the momentum equations in (\ref{eq:2.4}):
\begin{equation} \label{eq:2.60}
\begin{split}
&\rho \partial_tu+\rho u\cdot \nabla u+\nabla p_m=\frac{1}{c}\int_0^\infty \int_{S^{2}} K_a\cdot(I-\overline{B}(v))\Omega  \text{d}\Omega \text{d}v.
\end{split}
\end{equation}
That is to say, $(I,\rho,u)$ satisfies the Euler-Boltzmann equations classically. Then from the continuity equation, it is easy to get that $\rho$ can be expressed by
\begin{equation}
\label{eq:bb1}
\rho(t,x)=\rho_0(X(0,t,x))\exp\left(-\int_{0}^{t}\textrm{div} u(s,X(s,t,x))\text{d}s\right)\geq 0,
\end{equation}
where  $X\in C\left([0,T]\times[0,T]\times \mathbb{R}^3\right)$ is the solution of the initial value problem
\begin{equation}
\label{eq:bb2}
\begin{cases}
\frac{d}{dt}X(t,s,x)=u(t,X(t,s,x)),\quad 0\leq t\leq T,\\
X(s,s,x)=x, \quad \ \quad \quad 0\leq s\leq T,\ x\in \mathbb{R}^3,
\end{cases}
\end{equation}
In conclusion, Cauchy problem (\ref{eq:2.4})-(\ref{eq:2.211}) has a unique regular solution $(I,\rho,u)$.
%
%
\end{proof}

\section{Formation of singularities}\label{S3}
In this section, we consider the formation of singularities to regular solutions obtained in Section \ref{S2.3}. Two sufficient initial conditions that will lead to the finite time blow-up for  regular solutions will be given.  However, the second one is new even for compressible Euler equations or Euler equations with damping, which can be regarded as the multi-dimensional version of  \cite{lax} in the finite time blow-up sense when vacuum appears.

\subsection{\textbf{Local vacuum state}}\ \\

We first assume that the initial data $(I_0,\rho_0,u_0)$
satisfies the following local vacuum state condition:
\begin{definition}[\textbf{Local vacuum state}]\label{local}\ \\
Let $ A_0 $ and $ B_0 $ be two bounded open sets in $ \mathbb{R}^3 $, $B_0$ is connected and $\overline{A}_0 \subset B_0 \subseteq  B_{R_0}$, where $R_0$ is a positive constant and $B_{R_0}:=\{x \in \mathbb{R}^3:|x|\le R_0\}$. If the initial data $(I_0, \rho_0, u_0)$ satisfy
\begin{equation} \label{eq:12131}
\begin{cases}
\displaystyle
\rho_0(x)=u_0(x)=0, \ \forall x\in B_0-A_0;\ \int_{A_0} \rho_0(x) \text{d}x  =m_0 >0,\\[9pt]
\displaystyle
I_0\equiv \overline{B}(v),\  \ \forall (v,\Omega,x)\in \mathbb{R}^+\times S^{2}\times B^C_{R_0},
\end{cases}
\end{equation}
where $B^C_{R_0}=\{x\in \mathbb{R}^3| |x|\geq R_0\}$,
then we say that the initial data $(I_0, \rho_0, u_0)$ contain local vacuum state.
\end{definition}
\begin{remark}\label{r11}
$\overline{B}(v)$ is actually a simplification of the Planck function which represents the energy density of black-body radiation. Black-body has the smallest radiation, so condition
$I_0 \geq \overline{B}(v)$ is natural. In Theorem \ref{th33}, we will see that the assumption $ I_0\equiv \overline{B}(v)$ for $|x|\geq R_0$ results in the phenomenon that the impact of radiation on the dynamical properties of the fluid vanishes in the far field, then the system serves as the Euler equations as $|x|\rightarrow +\infty$.
\end{remark}

In order to observe the evolution of $A_0$ and $B_0$, we need the following definition.
\begin{definition}[\textbf{Particle path and flow map}]\label{kobn}\ \\
Let $x(t;0,x_0)$ be the particle path starting from $x_0$ at $t=0$, i.e.,
\begin{equation}\label{gobn}
\frac{d}{\text{d}t}x(t;0,x_0)=u(t,x(t;0,x_0)),\quad x(0;0,x_0)=x_0.
\end{equation}
Then we denote by $A(t)$, $B(t)$, $(B-A)(t)$ the images of $A_0$, $B_0$, and $B_0-A_0$, respectively, under the flow map of (\ref{gobn}), i.e.,
\begin{equation*}
\begin{split}
&A(t)=\left\{x(t;0,x_0)|x_0\in A_0\right\},\ B(t)=\left\{x(t;0,x_0)|x_0\in B_0\right\},\\
&(B-A)(t)=\left\{x(t;0,x_0)|x_0\in (B_0- A_0)\right\}.
\end{split}
\end{equation*}
It is easy to know that $(B-A)(t)$ is the vacuum domain.
\end{definition}

Then we have

\begin{lemma}
\label{lemma:3.1}
Let $(I,\rho,u)$ be the regular solution on $\mathbb{R}^+\times S^{2}\times[0,T)\times \mathbb{R}^3$ of the Cauchy problem (\ref{eq:2.4})-(\ref{eq:2.211}) satisfying (\ref{eq:12131}), then we have
\begin{equation}\label{tgb3}
B(t)= B_0, \ A(t)= A_0,\  \text{for} \  t \in [0,T).
\end{equation}
Moreover, there exists a critical time $T_c=\frac{2R_0}{c}$ such that, if $T> T_c$, then we have
\begin{equation}\label{tgb2} I(v,\Omega,t,x) \equiv \overline{B}(v),\quad (v,\Omega,t,x)\in \mathbb{R}^+\times S^{2}\times[T_c,T)\times B_0.\end{equation}
\end{lemma}
\begin{proof}
Firstly, on the domain $ (B-A)(t)$, $\rho=\overline{K}_a(v,t,x,\rho)\equiv 0$.
Due to the momentum equations in (\ref{eq:ccc}) and the definition of  regular solutions, we have
\begin{equation}\label{eq:5.3}
\partial_tu+u\cdot\nabla u=0, \quad \text{in} \quad   (B-A)(t).
\end{equation}
That is to say, $u$ is invariant along the particle path.
Thus, according to  the local vacuum state condition, we have
\begin{equation*}
u(t,x)\equiv 0, \quad \text{in} \quad  (B-A)(t).
\end{equation*}
Using the  continuity of $u(t,x)$, we get
\begin{equation*}
\frac{d}{\text{d}t}x(t;0,x_0)=u(t,x(t;0,x_0))\equiv 0 ,  \quad x_0 \in \partial B_0 \bigcup \partial A_0,
\end{equation*}
so $x(t;0,x_0)\equiv x_0$. Thus  $B(t)= B_0, \ A(t)= A_0$.

Secondly,
because  $\overline{B}(v)$ is independent of $x$ and $t$, the first equation of system (\ref{eq:2.4}) can be rewritten as
$$ \frac{1}{c} \partial_t(I-\overline{B}(v))+\Omega\cdot\nabla (I-\overline{B}(v))=-K_a\cdot(I-\overline{B}(v)).$$
We denote by $ y(t;y_0)$ the photon path starting from $y_0$ at $t=0$, i.e.,
      $$   \frac{d}{\partial{t}}y(t;y_0)=c\Omega ,     \qquad  y(0;y_0)= y_0 .                   $$
Along the photon path, we  obtain
     \begin{equation} \label{eq:**}(I-\overline{B}(v))(t,y(t;y_0 ))=(I_0-\overline{B}(v))(y_0 )\exp\Big(\int_0^t -cK_a(v,\tau,y(\tau;y_0 ),\rho)\text{d}\tau\Big),   \end{equation}
where $y_0=y-c\Omega t $.

Then  for our critical time $T_c=\frac{2R_0}{c}$, if $T> T_c$, via (\ref{tgb3}), we immediately have
$$ 
|y_0|=|y-c\Omega t|\geq R_0,\quad \text{for}\ (v,\Omega,t,y)\in \mathbb{R}^+\times S^{2}\times[T_c,T)\times B_0.$$
Due to $ I_0(v,\Omega,x) \equiv \overline{B}(v)$  for $ |y_0|\geq R_0$ and  (\ref{eq:**}), we deduce that 
 \begin{equation*} I(v,\Omega,t,y) \equiv \overline{B}(v),\quad (v,\Omega,t,y)\in \mathbb{R}^+\times S^{2}\times[T_c,T)\times B_0.\end{equation*}


\end{proof}

Now we give the main result of this section, which shows the formation of singularities caused by the appearance of vacuum in some local domain. We first introduce the mass and second moment over $B(t)$:
\begin{align*}
&m(t)=\int_{B(t)}\rho(t,x)\text{d}x \quad           (\text{mass}),\\
&M(t)=\int_{B(t)} \rho(t,x)|x|^{2}\text{d}x \quad                     (\text{second \ moment}).
\end{align*}

\begin{theorem}[\textbf{Finite time blow-up $1$}]\label{th33}\ \\ Assume that $(I,\rho,u)$
 is the regular solution on $\mathbb{R}^+\times S^2\times[0,T)\times \mathbb{R}^3$ of the Cauchy problem (\ref{eq:2.4})-(\ref{eq:2.211}) satisfying (\ref{eq:12131}),  then  it will
blow up in  finite time, i.e., $T<+\infty$.
\end{theorem}
\begin{proof}According to Lemma \ref{lemma:3.1}, we know that $B(t)=B_0$. So we easily have
\begin{equation}\label{zhi}
m(t)=m_0,\ \forall \ t\in[0,T).
\end{equation}
From the continuity equation and integration by parts, we have
\begin{equation}\label{eq:3.5}
\frac{d}{\text{d}t}M(t)=2\int_{B_0}x\cdot \rho u \text{d}x.
\end{equation}
From the momentum equations and integration by parts, we get
\begin{equation}\label{eq:3.6}
\frac{d^2}{\text{d}t^2}M(t)=2\int_{B_0}( \rho|u|^2 +3 p_m )\text{d}x+\frac{2}{c}\int_{B_0}\int_0^\infty \int_{S^{2}} K_a\cdot(I-\overline{B}(v))x\cdot\Omega \text{d}\Omega \text{d}v\text{d}x
\end{equation}

If $T\leq T_c$, the proof is over. So next we only consider the case that $T>T_c$. From Lemma \ref{lemma:3.1}, we know that
 \begin{equation}\label{liang} I(v,\Omega,t,x) \equiv \overline{B}(v),\quad (v,\Omega,t,x)\in \mathbb{R}^+\times S^{2}\times[T_c,T)\times B_0.\end{equation}
It follows from (\ref{eq:3.6})-(\ref{liang}) that
\begin{equation} \label{eq:3.8}\frac{d^2}{\text{d}t^2}M(t)\geq 2\int_{B_0}( \rho|u|^2 +3p_m) \text{d}x \geq 6\int_{B_0} p_m \text{d}x,\quad \text{for} \quad t\in [T_c,T).\end{equation}
From Holder's inequality, we give
\begin{equation}\label{eq:3.9}
m_0=\int_{B_0} \rho(t,x) \text{d}x  \leq \Big(\int_{|x|\leq R_0} \rho^\gamma(t,x) \text{d}x\Big)^{\frac{1}{\gamma}}\Big(\int_{|x|\leq R_0}\text{d}x\Big)^{\frac{1}{\gamma'}},\quad \text{for}\quad t\in [0,T).
\end{equation}
where $\frac{1}{\gamma}+\frac{1}{\gamma'}=1$. Associated with (\ref{zhi}), we give
\begin{equation}\label{eq:3.11}
\int_{|x|\leq R_0} p_m(t,x) \text{d}x \geq m^{\gamma}_0R^{3(1-\gamma)}_0|B_1|^{1-\gamma},\quad \text{for}\quad t\in [0,T).
\end{equation}
Then (\ref{eq:3.8}) yields
\begin{equation}\label{eq:3.12}
\frac{d^2}{\text{d}t^2}M(t)\geq 6 m^{\gamma}_0R^{3(1-\gamma)}_0|B_1|^{1-\gamma}, \quad \text{for} \quad t\in [T_c,T).
\end{equation}
So, using Taylor's expansion, we have
\begin{equation}\label{eq:3.13}
M(t)\geq M(0)+M'(0)t+3m^{\gamma}_0 R^{3(1-\gamma)}_0|B_1|^{1-\gamma} t^2, \quad \text{for} \quad t\in [T_c,T).
\end{equation}
From Lemma \ref{lemma:3.1}, it is clear that
\begin{equation}\label{kkk6}M(t)\leq   m_0 R^2_0,  \ \forall \ t\in[0,T).\end{equation}
Combining (\ref{eq:3.8})-(\ref{kkk6}), we have
\begin{equation}\label{eq:3.14}
m_0 R^2_0 \geq M(0)+M'(0)t+3m^{\gamma}_0 R^{3(1-\gamma)}_0|B_1|^{1-\gamma} t^2 , \quad \text{for} \quad t\in [T_c,T).
\end{equation}
Solving this inequality, we get
\begin{equation}\label{eq:3.15}
T_c\leq t \leq \frac{-M'(0)+\sqrt{M'(0)^2-12m^{\gamma}_0 R^{3(1-\gamma)}_0|B_1|^{1-\gamma} (M(0)- m_0 R^2_0)}}{6m^{\gamma}_0 R^{3(1-\gamma)}_0|B_1|^{1-\gamma} }.
\end{equation}
In other words, the life span $T$ must be finite.
\end{proof}

\begin{remark}
Theorem \ref{th33} stated that  the appearance of vacuum will cause the blow-up of the regular solutions of Euler-Boltzmann equations in finite time. However, this kind of  singularity is different from the shock wave which is caused by the compression of fluid. The corresponding results for Euler equations, damped Euler equations, and Euler-Possion equations can be found in  \cite{tpy}, \cite{tpl2}, \cite{tms3}, \cite{tbd}, \cite{ts2}, etc. Moreover, the result obtained in Theorem \ref{th33} also improved the conclusion in \cite{tpy}\cite{tms3}\cite{tms1} in the sense that we removed the crucial assumption that the initial mass density is compactly supported.
\end{remark}

\subsection{\textbf{Hyperbolic singularity set}}\ \\

%
%
In this section, we will show our second type initial condition that will lead to the finite time blow-up for our regualr solution obtained in Section $2$. 
We first give the defintion of the hyperbolic singularity set:

\begin{definition}[\textbf{Hyperbolic singularity set}]\label{bugers}
We define the smooth,  open set $V \subset \Omega$ as a hyperbolic singularity set, if $V$ and $(\rho_0,u_0)$ satisfy
\begin{equation} \label{eq:12131sss}
\begin{cases}
\displaystyle
\rho_0(x)=0, \ \forall \ x\in V;\\[10pt]
\displaystyle
 Sp(\nabla u_0) \cap \mathbb{R}^-\neq\  \emptyset,\quad  \forall  \ x \in V,
\end{cases}
\end{equation}
where we denote by $Sp(\nabla u_0(x))$ the spectrum of the Jacobian matrix of $u_0$.
\end{definition}
Then we show that 

\begin{theorem}[\textbf{Finite time blow-up $2$}]\label{th333}\ \\ Assume that $(I,\rho,u)$
 is the regular solution on $\mathbb{R}^+\times S^2\times[0,T)\times \mathbb{R}^3$ of the Cauchy problem (\ref{eq:2.4})-(\ref{eq:2.211}) satisfying (\ref{eq:12131sss}),  then  it will
blow up in  finite time, i.e., $T<+\infty$.
\end{theorem}
\begin{proof}
We denote by $ V(t) $ the evolutioned domain that is the image of $ V$ under the flow map, i.e.,
\begin{equation}\label{zhi}
V(t)=\{x|x=x(t;0,\xi_0), \quad \forall \xi_0\  \in V\},
\end{equation}
where $ x(t; 0,\xi_0)$ is the particle path starting from $\xi_0$ when $t=0$, namely,
\begin{equation}\label{particle}  \frac{d}{\text{d}t}x(t;0, \xi_0)=u(t,x(t;0, \xi_0)),   \quad x(0; 0,\xi_0)= \xi_0.
\end{equation}
      It follows from the continuity equation that the smooth solution is simply supported along the particle paths, so
       $$ \rho(t,x)=0,\quad \text{when} \quad x\ \in \ V(t).$$
Thus, via the Definition \ref{666} for regular solutions, we deduce that 
\begin{equation}
\label{eq:1.2guo}
\partial_t u+u\cdot \nabla u=0, \quad \text{when} \quad x\ \in \ V(t),
\end{equation}
which means that $u$ is a constant along the particle path $x(t; 0,\xi_0)$.
Then for any $x\in V(t)$, we obtain that 
$$
u(t,x)=u_0(x-tu(t,x)),
$$
which immediately implies that 
\begin{equation}
\label{eq:1.2fan}
\nabla u(t,x)=\big(\mathbb{I}_3+t\nabla u_0(x-tu(t,x))\big)^{-1}\nabla u_0,\quad \text{for} \quad x\in V(t).
\end{equation}
If there is any $\lambda\in  Sp(\nabla u_0) $ satisfying $\lambda<0$, 
 then from (\ref{eq:1.2fan}), it is obvious that the quantity $\nabla u$ will blow up in finte time, i.e., 
$$
T<+\infty.
$$

\end{proof}

Then we immediately have the following corollaries for compressible isentropic Euler equations and Euler equations with damping:
\begin{corollary}[\textbf{Finite time blow-up $3$}]\label{th3333}\ \\ 
Let $I\equiv \overline{B}(v)$. Assume that $(\rho,u)$
 is the regular solution on $[0,T)\times \mathbb{R}^3$ of the Cauchy problem (\ref{eq:2.4})-(\ref{eq:2.211}) satisfying (\ref{eq:12131sss}),  then  it will
blow up in  finite time, i.e., $T<+\infty$.
\end{corollary}

The compressible isentropic Euler equations with damping can be given as
\begin{equation}
\label{eq:1.2sss}
\begin{cases}
\displaystyle
\partial_{t}\rho+\nabla\cdot(\rho u)=0,\\[10pt]
\displaystyle
\partial_{t}(\rho u)+\nabla\cdot(\rho u\otimes u)
  +\nabla p_m =-\alpha \rho u,
\end{cases}
\end{equation}
where $\alpha>0$. The corresponding blow-up results can be written as:
\begin{corollary}[\textbf{Finite time blow-up $4$}]\label{th33333}\ \\ 
 Assume that $(\rho,u)$
 is the regular solution on $[0,T)\times \mathbb{R}^3$  obtained in \cite{tpy} to the Cauchy problem (\ref{eq:1.2sss}) with initial data $(\rho_0,u_0)$. If $(\rho_0,u_0)$ satisfies that 
\begin{equation} \label{eq:121331sss}
\begin{cases}
\displaystyle
\rho_0(x)=0, \ \forall \ x\in V;\\[10pt]
\displaystyle
 Sp(\nabla u_0) \cap \{\lambda\in \mathbb{R}|\lambda<-\alpha\}\neq\  \emptyset,\quad  \forall  \ x \in V,
\end{cases}
\end{equation}
 then  it will
blow up in  finite time, i.e., $T<+\infty$.
\end{corollary}
\begin{proof}
Similarly to the proof of Theorem \ref{th333}, and  the definition  for regular solutions in \cite{tpy}, we deduce that 
\begin{equation}
\label{eq:1.2guo1}
\partial_t u+u\cdot \nabla u=-\alpha u, \quad \text{when} \quad x\ \in \ V(t).
\end{equation}
Let 
$$
\overline{u}(t,x)=e^{\alpha t}u(t,x),$$
then it is easy to see that 
\begin{equation}
\label{bianxing1}
\partial_t \overline{u}+u\cdot \nabla \overline{u}=0, \quad \text{when} \quad x\ \in \ V(t),
\end{equation}
which implies that  along  the particle path  $ x(t; t,x)$ (see (\ref{particle})),  $\overline{u}(t,x)$ is a constant and 
\begin{equation}\label{changxi}
\overline{u}(t,x)=u_0(\xi_0).
\end{equation}
And, according to the (\ref{particle}), $\xi_0$ satisfies that 
\begin{equation}\label{chuzhi}\begin{split}
x-\xi_0=&\int_0^t u(\tau;t,x) \text{d} \tau
= \int_0^t e^{-\alpha \tau} \overline{u}(\tau;t,x) \text{d}\tau\\
=&-\frac{1}{\alpha} \overline{u}(t,x) (e^{-\alpha t}-1).
\end{split}
\end{equation}
Then from (\ref{changxi})-(\ref{chuzhi}) it is easy to see that 
\begin{equation}
\label{eq:1.2fan1}
\nabla \overline{u}(t,x)=\Big(\mathbb{I}_3-\frac{1}{\alpha}(e^{-\alpha t}-1)\nabla u_0(\xi_0)\Big)^{-1}\nabla u_0,\quad \text{for} \quad x\in V(t).
\end{equation}
If there is any $\lambda\in  Sp(\nabla u_0) $ satisfying $\lambda<-\alpha$,
 we need to  consider the functions 
$$f(t)=1-\frac{1}{\alpha}(e^{-\alpha t}-1)\lambda.$$ 
We easily have
$$
f(0)=1, \quad  \quad f(+\infty)=1+\frac{\lambda}{\alpha}<0,
$$
which implies that there must exists some finite time $t_0$ such that 
$$
f(t_0)=0,
$$
then from (\ref{eq:1.2fan1}), it is obvious that the quantity $\nabla u$ will blow up in finte time, i.e., 
$$
T<+\infty.
$$

\end{proof}
\begin{remark}\label{zheng9}
The finite time blow-up results obtained in Theorems \ref{th33}-\ref{th333} and Corollaries  \ref{th3333}-\ref{th33333} can be easily generalized to the $1$D space.  And we compare Corollary \ref{th3333} in $1$D case with the formation of singularities obtained in \cite{lax}.  If the initial mass density $\rho_0$ has a positive  lower bound,  in order to obtain the finite blow-up, in \cite{lax} we need to assume that there is some point $x\in \mathbb{R}$ such that 
$$(w_0)_x=(u_0)_x-C(\rho^{\frac{\gamma-1}{2}}_0)_x<0,$$
 where $w$ is the Riemann invariant. However, this result cannot be extended to the multi-dimensional space due to the existence of the Riemann invariant and  other difficulties.

In our theorem, for  $1$D case, if the initial mass density vanished in some open set $V$, in order to achieving the same target, we need to assume that  there is some point $x\in V$ such that 
$$(w_0)_x=(u_0)_x-C(\rho^{\frac{\gamma-1}{2}}_0)_x<0.$$
And we can generalized this condition to the multi-dimensional space such as
\begin{equation*} 
\begin{cases}
\displaystyle
\rho_0(x)=0, \ \forall \ x\in V;\\[10pt]
\displaystyle
 Sp(\nabla u_0) \cap \mathbb{R}^-\neq\  \emptyset,\quad  \forall  \ x \in V,
\end{cases}
\end{equation*}
for compressible Euler equations or Euler-Boltzmann equations.
\end{remark}

\section{Local Existence for the case $\sigma_s\neq 0$}

In this section, we will give the corresponding  local existence of regular solutions for the case $\sigma_s\neq 0$, which is similar to the result obtained in Section $2$, and we use the same notation as in Section $2$. When $\sigma_s\neq 0$, if we still consider the assumptions of `induced process' and local thermal equilibrium as in Section $2$ for the case $\sigma_s= 0$, then the Euler-Boltzmann equations (\ref{eq:1.2}) become very complicated, for example, the radiation transfer equation in (\ref{eq:1.2}) reads as
$$
\frac{1}{c}\partial_tI+\Omega\cdot\nabla I=-K_a\cdot(I-\overline{B}(v))+\Pi,
$$
where
$$
\Pi=\int_0^\infty \int_{S^{2}} \left(\frac{v}{v'}\sigma_sI'\left(1+\frac{c^{2}I}{2hv^3}\right)
-\sigma'_sI\left(1+\frac{c^{2}I'}{2hv^3}\right)\right) \text{d}\Omega' \text{d}v',
$$
which is rather complicated and hard to deal with.
Therefore, for simplicity, we start from the original Euler-Boltzmann equations (\ref{eq:1.2}). For this, we need some assumptions. Let
$$ \sigma_s=\rho\overline{\sigma}_s(v' \rightarrow v, \Omega'\cdot\Omega),\quad \sigma'_s=\rho\overline{\sigma}'_s(v \rightarrow v', \Omega\cdot\Omega'),$$
where $\overline{\sigma}_s\geq 0$ and $\overline{\sigma}'_s\geq 0$ satisfy
\begin{equation}\label{zhen8}
\begin{split}
&\int_0^\infty \int_{S^{2}}\left(\int_{0}^{\infty}\int_{S^{2}} \Big|\frac{v}{v'}\Big|^2|\overline{\sigma}_s|^2\text{d}\Omega' \text{d}v'\right)^{\lambda_1}\text{d}\Omega \text{d}v\leq C,\\
&\int_0^\infty \int_{S^{2}}\left(\int_{0}^{\infty}\int_{S^{2}} \overline{\sigma}'_s\text{d}\Omega' \text{d}v'\right)^{\lambda_2}\text{d}\Omega \text{d}v+\int_{0}^{\infty}\int_{S^{2}} \overline{\sigma}'_s\text{d}\Omega' \text{d}v'\leq C,
\end{split}
\end{equation}
where $\lambda_1=1$ or $\frac{1}{2}$, and $\lambda_2=1$ or $2$.
Let
\begin{equation} \label{fg11}
\begin{cases}
S=S(v,t,x,\rho)=\rho \overline{S}(v,t,x,\rho)=\rho \overline{S},\quad \overline{S}\geq 0, \\[9pt]
\sigma_a=\sigma_a(v,t,x,\rho)=\rho \overline{\sigma}_a(v,t,x,\rho)=\rho\overline{\sigma}_a,\quad \overline{\sigma}_a\geq 0,
\end{cases}
\end{equation}
and for any $\|w(t,\cdot)\|_s\leq M$ (here $ w=\rho^{\frac{\gamma-1}{2}}$, $M$ is a positive constant ), we assume
\begin{equation} \label{fg22}
\begin{cases}
\displaystyle
\|\overline{S}\|_{L^1(\mathbb{R}^+;C([0,T];H^s))}+\|S\|_{L^2(\mathbb{R}^+;C([0,T];H^s))}\leq C_{s,M}\|w\|_s,\\[9pt]
 \|\sigma_a\|_{L^\infty(\mathbb{R}^+;C([0,T];H^s))}+\| \overline{\sigma}_a \|_{L^2(\mathbb{R}^+;C([0,T];H^s))}\leq C_{s,M}\|w\|_s,\\[9pt]
\displaystyle
\|(|\partial_w\overline{\sigma}_a|+|\partial_w\overline{S}|)(v,t,x,w)\|_{L^\infty \cap L^1(\mathbb{R}^+;L^\infty(\mathbb{R}^3))}\leq C_{s,M},\ \forall \ t\in [0,T].
\end{cases}
\end{equation}
\begin{remark}The evaluation of these physical coefficients is a difficult problem in quantum mechanics and their general form is not known.
Some  similar assumptions and physical examples on $S$, $\sigma_a$ and $\sigma_s$ can be found in Remarks $2.2$ and $3.1$ of \cite{sjx} as well as in \cite{gp}. When $\gamma\leq 1+\frac{2}{s}$, that is $\frac{2}{\gamma-1}\geq s$,  (\ref{fg22}) is similar to the assumptions in \cite{sjx} for the local existence of classical solutions to Cauchy problem (\ref{eq:1.2}) and (\ref{eq:2.211}) with initial mass density away from vacuum.
\end{remark}

Similarly, in order to change (\ref{eq:1.2}) into a symmetric hyperbolic system, we introduce the new variable
$$ w=p^{\frac{\gamma-1}{2\gamma}}_m=\rho^{\frac{\gamma-1}{2}},$$
and denote $U=U(t,x)=(w(t,x),u(t,x))^\top$. Then system (\ref{eq:1.2}) of the isentropic Euler-Boltzmann equations
can be reduced to the following system:
\begin{equation}\label{eq:2.6tt}
\begin{cases}
\displaystyle
\frac{1}{c}\partial_tI+\Omega\cdot\nabla I=A_r,\\[10pt]
\displaystyle
 A_{0}(U)\partial_tU+\sum_{j=1}^{3} A_{j}(U)\partial_{x_j}U=F(I,U),
\end{cases}
\end{equation}
where  $A_0(U)$ and $A_j(U)$ are defined in Section $2$, and
\begin{equation*}
\begin{split}
&F(I,U)=(F_0,F_1,F_2,F_3),\ F_0(I,U)=0,\\
&F_j(I,U)=-\frac{(\gamma-1)^2}{4c\gamma}\int_0^\infty \int_{S^{2}} \Big(\overline{S}-\overline{\sigma}_aI
+\int_0^\infty \int_{S^{2}} \big(\frac{v}{v'}\overline{\sigma}_sI'-\overline{\sigma}'_sI\big)
 \text{d}\Omega' \text{d}v'\Big)\Omega_j\text{d}\Omega \text{d}v,
\end{split}
\end{equation*}
where $j=1,2,3$. Similarly to Theorem \ref{th:1}, in order to get the local existence of the original Cauchy problem (\ref{eq:1.2}) and (\ref{eq:2.211}), we need the following key theorem.
\begin{theorem}\label{th:1rr}
Let $s\geq 3$ be an integer and  (\ref{fg11})-(\ref{fg22}) hold.
If the initial data satisfy
\begin{equation*}\begin{split}
(I_0,U_0) \in \Psi:&=\left\{(I,U)|U(x)\in H^{s}(\mathbb{R}^3), \ I(v,\Omega,x) \in L^2(\mathbb{R}^+\times S^{2};H^{s}(\mathbb{R}^3))\right\},
\end{split}
\end{equation*}
then there exists $T > 0$ such that the problem (\ref{eq:2.6tt}) and (\ref{eq:2.8}) has a unique classical solution $(I,U)$ satisfying
\begin{equation*}
U\in C^1\left([0,T)\times \mathbb{R}^3\right),\ I\in L^2\left(\mathbb{R}^+\times S^2; C^1([0,T)\times \mathbb{R}^3)\right).
\end{equation*}
\end{theorem}

We can follow the same procedure as the proof of Theorem \ref{th:1} to prove Theorem \ref{th:1rr}. For $k=0,1,2,...$, we define $U^{(k+1)}(t,x)$ and $I^{(k+1)}(v,\Omega,t,x)$ inductively as the solution of the following linearized problem:
\begin{equation}\label{eq:2.13tt}
\begin{cases}
\displaystyle
\frac{1}{c}\partial_tI^{(k+1)}+\Omega\cdot\nabla I^{(k+1)}+\left(\sigma^{(k)}_a+\int_{0}^{\infty}\int_{S^{2}} \sigma'^{(k)}_s\text{d}\Omega' \text{d}v'\right)I^{(k+1)}=A^{(k)}_r,\\[8pt]
\displaystyle
A_{0}(U^{(k)})\partial_tU^{(k+1)}+\sum_{j=1}^{3} A_{j}(U^{(k)}) \partial_{x_{j}}U^{(k+1)}=F(I^{(k)},U^{(k)}),\\[8pt]
I^{(k+1)}(0,x,v,\Omega)=I^{(k+1)}_0(x,v,\Omega), \ U^{(k+1)}(0,x)=U^{(k+1)}_0(x),
\end{cases}
\end{equation}
where
\begin{equation}\label{jihaott}
\begin{cases}
\displaystyle
A^{(k)}_r=S^{(k)}
+\int_0^\infty \int_{S^{2}}\frac{v}{v'}\sigma^{(k)}_sI'^{(k)}
 \text{d}\Omega' \text{d}v',\quad F_0(I^{(k)},U^{(k)})=0,\\[8pt]
\displaystyle
F_j(I^{(k)},U^{(k)})=-\frac{(\gamma-1)^2}{4c\gamma}\int_0^\infty \int_{S^{2}} \left(\overline{S}^{(k)}
-\overline{\sigma}^{(k)}_aI^{(k+1)}\right)\Omega_j\text{d}\Omega \text{d}v,\\[8pt]
\displaystyle
\qquad -\frac{(\gamma-1)^2}{4c\gamma}\int_0^\infty \int_{S^{2}} \left(\int_0^\infty \int_{S^{2}} \left(\frac{v}{v'}\overline{\sigma}_sI'^{(k)}-\overline{\sigma}'_sI^{(k+1)}\right)
 \text{d}\Omega' \text{d}v'\right)\Omega_j\text{d}\Omega \text{d}v,\\[8pt]
 \sigma^{(k)}_s=\sigma_s(v' \rightarrow v,\Omega'\cdot\Omega,w^{(k)}), \  \sigma'^{(k)}_s=\sigma_s(v \rightarrow v',\Omega\cdot\Omega',w^{(k)}),\\[8pt]
S^{(k)}=S(v,t,x,w^{(k)}),\ \overline{S}^{(k)}=\overline{S}(v,t,x,w^{(k)}),\\[8pt]
 \sigma^{(k)}_a=\sigma_a(v,t,x,w^{(k)}),\ \overline{\sigma}^{(k)}_a=\overline{\sigma}_a(v,t,x,w^{(k)}).
\end{cases}
\end{equation}
It follows immediately that
\begin{equation}\label{eq:2.16tt}
U^{(k+1)}\in C^\infty\left([0,T_k]\times \mathbb{R}^3\right),\ I^{(k+1)}\in L^2\left(\mathbb{R}^+\times S^2,\ C^\infty([0,T_k]\times \mathbb{R}^3)\right),
\end{equation}
where $T_k$ is the largest time of existence for (\ref{eq:2.13tt}) such that the estimates
\begin{equation}\label{eq:2.17tt}
 \int_{0}^{\infty}\int_{S^{2}} |\|I^{(k)}-I^{(0)}_0\||^2_{s,T_k} d\Omega \leq C,\ |\|U^{(k)}-U^{(0)}_0\||_{s,T_k}\leq C
\end{equation}
are valid for any given constant $C>0$. Next we give two key lemmas as in Section $2$, which imply the compactness of the above-constructed approximate solutions.
\begin{lemma}[\textbf{Boundedness in the high norm}]  \label{lemma:2:1tt}\ \\
There exist constants $C_2 > 0$  and $\overline{T}_* > 0$ such that  the solution $(I^{(k)},U^{(k)})$ to (\ref{eq:2.13tt}) and (\ref{eq:2.8}) satisfies
\begin{equation}\label{eq:2.19tt}
|\|U^{(k)}-U^{(0)}_0\||_{s,\overline{T}_*}+|\|\partial_tU^{(k)}\||_{s-1,\overline{T}_*}+ \int_{0}^{\infty}\int_{S^{2}} |\|I^{(k)}-I^{(0)}_0\||^2_{s,  \overline{T}_*} d\Omega \text{d}v\leq C_2,
\end{equation}
for $k=0,1,2,....$
\end{lemma}
\begin{proof}It is sufficient to prove that (\ref{eq:2.19tt}) holds for $(I^{(k+1)},U^{(k+1)})$ under the assumption that (\ref{eq:2.19tt}) holds for $(I^{(k)},U^{(k)})$. We divide the proof into three steps.

\underline{Step 1}. The estimate of $ \int_{0}^{\infty}\int_{S^{2}} |\|I^{(k+1)}-I^{(0)}_0\||^2_{s,T} \text{d}\Omega \text{d}v $.

Let $ Q^{(k+1)}=I^{(k+1)}-I^{(0)}_0 $.  Then $ Q^{(k+1)}$ satisfies
\begin{equation} \label{eq:2.21tt}
\begin{cases}
\displaystyle
\frac{1}{c}\partial_tQ^{(k+1)}+\Omega\cdot\nabla Q^{(k+1)}+\Big(\sigma^{(k)}_a+\int_{0}^{\infty}\int_{S^{2}} \sigma'^{(k)}_s\text{d}\Omega' \text{d}v'\Big)Q^{(k+1)}
=S^{(k)}+\Theta^{(k)},\\
Q^{(k+1)}(v,\Omega,0,x)=I^{(k+1)}_0-I^{(0)}_0 ,
\end{cases}
\end{equation}
where
\begin{equation} \label{eq:2.21mtt}
\Theta^{(k)}=-\Big(\sigma^{(k)}_a+\int_{0}^{\infty}\int_{S^{2}} \sigma'^{(k)}_s\text{d}\Omega' \text{d}v'\Big) I^{(0)}_0 -\Omega \cdot \nabla I^{(0)}_0+\int_0^\infty \int_{S^{2}} \frac{v}{v'}\sigma^{(k)}_sI'^{(k)}
 \text{d}\Omega' \text{d}v'.
\end{equation}
Differentiating the equations in (\ref{eq:2.21tt}) $\alpha$-times ($|\alpha|\leq s$) with respect to $x$, from Young's inequality we have
\begin{equation} \label{eq:2.24tt}
\begin{split}
&\frac{d}{\text{d}t}\int_{\mathbb{R}^3}|D^{\alpha}Q^{(k+1)}|^2dx
\leq  C \|D^{\alpha}Q^{(k+1)}\|^2
+\|D^{\alpha}S^{(k)}\|^2 +\|D^{\alpha}\Theta^{(k)}\|^2\\
&+\|D^{\alpha}(\sigma^{(k)}_aQ^{(k+1)})\|^2+\int_0^\infty \int_{S^{2}} \int_{\mathbb{R}^3} D^\alpha(\sigma'^{(k)}_sQ^{k+1})
D^\alpha Q^{(k+1)}\text{d}x \text{d}\Omega'\text{d}v'\\
=&:\overline{J}_1+\overline{J}_2+\overline{J}_3+\overline{J}_4+\overline{J}_5.
\end{split}
\end{equation}
According to Lemma \ref{pro:2.3}, we get
\begin{equation} \label{eq:2.26tt}
\begin{split}
\overline{J}_3=\|D^{\alpha}\Theta^{(k)}\|^2
\leq &\left\|D^{\alpha}\left(-\sigma^{(k)}_aI^{(0)}_0 -\int_{0}^{\infty}\int_{S^{2}} \sigma'^{(k)}_s\text{d}\Omega' \text{d}v' I^{(0)}_0-\Omega \cdot \nabla I^{(0)}_0\right) \right\|^2\\
&+\left\|\int_0^\infty \int_{S^{2}}  D^\alpha\left(\frac{v}{v'}\sigma^{(k)}_sI'^{(k)}
\right) \text{d}\Omega'\text{d}v'\right\|^2\\
\leq&  C \|\sigma^{(k)}_a\|^2_{s}\|I^{(0)}_0\|^2_s +\left(1+\|U^k\|^2_s\right)\| I^{(0)}_0\|_{s+1}^2+\Delta,
\end{split}
\end{equation}
where
\begin{equation} \label{eq:2.27tt}
\begin{split}
\Delta=&\left\|\int_0^\infty \int_{S^{2}}  D^\alpha\left(\frac{v}{v'}\sigma^{(k)}_sI'^{(k)}
\right) \text{d}\Omega'\text{d}v'\right\|^2 \leq C\left( \int_{0}^{\infty} \int_{S^{2}} \Big|\frac{v}{v'}\Big|
\|\sigma^{(k)}_s \|_s \|I'^{(k)}\|_s \text{d}\Omega'\text{d}v' \right)^2 \\
\leq & C\|U^{(k)}\|_s  \int_{0}^{\infty} \int_{S^{2}}\|I^{(k)}\|^2_s \text{d}\Omega \text{d}v \cdot \int_0^\infty  \int_{S^{2}}  \frac{v^2}{v'^2}  |\overline{\sigma}_s |^2\text{d}\Omega'\text{d}v',
\end{split}
\end{equation}
where we used Holder's inequality, Minkowski's inequality and assumptions (\ref{zhen8})-(\ref{fg22}). We also have
\begin{equation} \label{eq:2.27ttt}
\begin{split}
\overline{J}_4=&\|D^{\alpha}(\sigma^{(k)}_aQ^{(k+1)})\|^2\leq C\|\sigma^{(k)}_a\|^2_s\|Q^{(k+1)}\|^2_s\leq C\|U^k\|^2_s\|Q^{(k+1)}\|^2_s,\\
\end{split}
\end{equation}
and
\begin{equation} \label{eq:2.27tttt}
\begin{split}
\overline{J}_5=&\int_0^\infty \int_{S^{2}} \int_{\mathbb{R}^3} D^\alpha(\sigma'^{(k)}_sQ^{k+1})
D^\alpha Q^{(k+1)}\text{d}x \text{d}\Omega'\text{d}v'\leq C\|U^k\|_s\|Q^{(k+1)}\|^2_s.
\end{split}
\end{equation}

Combining (\ref{zhen8})-(\ref{fg22}) and  (\ref{eq:2.24tt})-(\ref{eq:2.27tttt}) , we have
\begin{equation*}\begin{split}
\frac{d}{\text{d}t}\|Q^{(k+1)}\|_s^2
\leq&  C\left(\|D^{\alpha}S^{(k)}\|^2+\|Q^{(k+1)}\|_s^2
+\|I^{(0)}_0\|^2_{s+1} +\int_0^\infty  \int_{S^{2}}  \frac{v^2}{v'^2}  \| \overline{\sigma}_s \|^2 \text{d}\Omega'\text{d}v' \right).
\end{split}
\end{equation*}
Integrating this inequality with respect to $t$ over $[0,T]$, according to Gronwall's inequality, we obtain
$$
\int_{0}^{\infty}\int_{S^{2}}|\|Q^{(k+1)}\||^2_{s, T}\text{d}\Omega \text{d}v\leq C e^{(C_2+1) T}\left(\int_{0}^{\infty}\int_{S^{2}}\|Q^{(k+1)}_0\|_s^2\text{d}\Omega \text{d}v+T\right).$$
Taking $\overline{T}_1$ to be enough small, we arrive at
\begin{equation}\label{kkatt}
\int_{0}^{\infty}\int_{S^{2}} |\|I^{(k+1)}-I^{(0)}_0\||^2_{s, \overline{T}_1}  \text{d}\Omega \text{d}v\leq C_2.
\end{equation}

\underline{Step 2}. The estimate of source term $\|D^{\alpha}F(I^{(k)}, U^{(k)})\|$, $\forall|\alpha| \leq s $. 

Due to Minkowski's inequality, Holder's inequality and (\ref{pro:2.311}), for $|\alpha| \leq s $, we have
\begin{equation*}
\begin{split}
 &\| D^{\alpha}F_j(I^{(k)},U^{(k)}) \|\\
 \leq &
C
 \left\|D^\alpha \int_0^\infty \int_{S^{2}}   \left(\overline{S}^{(k)}-\overline{\sigma}^{(k)}_aI^{(k)}
+\int_0^\infty \int_{S^{2}} \Big(\frac{v}{v'}\overline{\sigma}_sI'^{(k)}-\overline{\sigma}'_sI^{(k+1)}
\Big) \text{d}\Omega' \text{d}v'\right)\Omega_j \text{d}\Omega \text{d}v \right\|\\
\leq & C \left( \int_{0}^{\infty} \int_{S^2} (\| \overline{S}^{(k)} \|_s + \| \overline{\sigma}^{(k)}_a I^{(k)}  \|_s) \text{d}\Omega \text{d}v+\overline{J}_6+\overline{J}_7 \right) \\
\leq & C  \left( \int_{0}^{\infty} \int_{S^2} \| \overline{S}^{(k)} \|_s\text{d}\Omega \text{d}v +\int_{S^{2}}\Big(\int_0^\infty \|\overline{\sigma}^{(k)}_a\|^2_s \text{d}v\Big)^{\frac{1}{2}}\Big( \int_0^\infty \|I^{(k)}\|_s^2 \text{d}v \Big)^{\frac{1}{2}} \text{d}\Omega+\overline{J}_6+\overline{J}_7 \right),
\end{split}
\end{equation*}
where
\begin{equation}\label{eq:2.32tt}
\begin{split}
\overline{J}_6=&\Big
\|\int_0^\infty \int_{S^{2}}\int_0^\infty \int_{S^{2}}\left(\frac{v}{v'}\overline{\sigma}_sI'^{(k)}\right)\Omega_j \text{d}\Omega' \text{d}v'\text{d}\Omega \text{d}v\Big\|_s\\
\leq &  C\int_{0}^{\infty} \int_{S^{2}}  \Big(\int_{0}^{\infty}\int_{S^{2}}\Big|\frac{v}{v'}\Big| | \overline{\sigma}_s |  \|I'^{(k)}\|_s \text{d}\Omega' \text{d}v' \Big) \text{d}\Omega \text{d}v\\
\leq& C\left(\int_{0}^{\infty}\int_{S^{2}}\|I^{(k)}\|^2_s \text{d}\Omega \text{d}v\right)^{\frac{1}{2}}
 \left(\int_0^\infty \int_{S^{2}}\Big(\int_{0}^{\infty}\int_{S^{2}} \Big|\frac{v}{v'}\Big|^2\overline{\sigma}^2_s \text{d}\Omega' \text{d}v'\right)^{\frac{1}{2}}\text{d}\Omega \text{d}v,
 \end{split}
 \end{equation}
and
 \begin{equation}\label{eq:2.32ttt}
\begin{split}
\overline{J}_7=&\Big
\|\int_0^\infty \int_{S^{2}}\int_0^\infty \int_{S^{2}}\overline{\sigma}'_sI^{(k+1)}\Omega_j \text{d}\Omega' \text{d}v'\text{d}\Omega \text{d}v\Big\|_s\\
\leq &  C\int_{0}^{\infty} \int_{S^{2}}  \|I^{(k+1)}\|_s \Big(\int_{0}^{\infty}\int_{S^{2}} | \overline{\sigma}'_s |  \text{d}\Omega' \text{d}v' \Big) \text{d}\Omega \text{d}v\\
\leq& C\left(\int_{0}^{\infty}\int_{S^{2}}\|I^{(k+1)}\|^2_s \text{d}\Omega \text{d}v\right)^{\frac{1}{2}}\left(\int_0^\infty \int_{S^2} \Big( \int_{0}^{\infty}\int_{S^2} \overline{\sigma}'_s\text{d}\Omega \text{d}v \Big)^{2}\text{d}\Omega' \text{d}v' \right)^{\frac{1}{2}}.
\end{split}
\end{equation}
Together with
  assumptions (\ref{fg11})-(\ref{fg22}), we obtain
\begin{equation}\label{eq:2.3eet}
\|F(I^{(k)},U^{(k)})\|_s \leq  C\Big(\|U^{(k)}\|_s+ \int_{0}^{\infty}\int_{S^{2}} \Big(\|I^{(k)}\|^2_s+\|I^{(k+1)}\|^2_s\Big) d\Omega \text{d}v\Big)\leq CC_2.
\end{equation}
\underline{Step 3}. In order to estimate (\ref{eq:2.19tt}), let $M^{(k+1)}=U^{(k+1)}-U^{(0)}_0$. It is easy to get
\begin{equation}\label{eq:2.35tt}
\begin{cases}
\displaystyle
 A_{0}(U^{(k)})\partial_tM^{(k+1)}+\sum_{j=1}^{3} A_{j}(U^{(k)})\partial_{x_{j}}{M^{(k+1)}}=F(I^{(k)},U^{(k)})+\overline{\Theta}^{(k)},\\[10pt]
 M^{(k+1)}(0,x)=U^{(k+1)}_0(x)-U^{(0)}_0(x),
\end{cases}
\end{equation}
where
$$
\overline{\Theta}^{(k)}=-\sum_{j=1}^{3} A_{j}(U^{(k)})\partial_{x_{j}}{U^{(0)}_0}.
$$
With the aid of the steps $1$ and $2$, we can easily follow the standard procedure as in \cite{amj} to show that there is a time $\overline{T}_2$ satisfying
\begin{equation}\label{eq:2.36-1}
|\|U^{(k+1)}-U^{(0)}_0\||_{s,T_2}+|\|\partial_tU^{(k+1)}\||_{s-1,T_2}\leq C_2.
\end{equation}
Let $ \overline{T}_*=\min\{\overline{T}_1,\overline{T}_2\} $. Then Lemma \ref{lemma:2:1tt} is proved.
\end{proof}
\begin{lemma}\label{lemma:2:2tt}
There exist constants $\overline{T}_{\ast\ast}\in [0,\overline{T}_*]$, $\overline{\eta} < 1$ , $\{\overline{\beta}_k\} \ (k=1,2,...)$  and  $\{\overline{\mu}_k\} \ (k=1,2,...)$  with  $\sum_{k} |\overline{\beta}_k|< +\infty$ and
$\sum_{k} |\overline{\mu}_k|< +\infty $ , such that for each $k$
\begin{equation}\label{eq:2.37-1}
\begin{split}
&|\|U^{(k+1)}-U^{(k)}\||_{0,\overline{T}_{\ast\ast}}+\left(\int_0^\infty \int_{S^{2}}\||I^{(k+1)}-I^{(k)}\||^2_{0,\overline{T}_{\ast\ast}} \text{d}\Omega \text{d}v\right)^{\frac{1}{2}}\\
\leq &\overline{\eta} \left(|\|U^{(k)}-U^{(k-1)}\||_{0,\overline{T}_{\ast\ast}}+\left(\int_0^\infty \int_{S^{2}}\||I^{(k)}-I^{(k-1)}\||^2_{0,\overline{T}_{\ast\ast}} \text{d}\Omega \text{d}v\right)^{\frac{1}{2}}\right)+\overline{\beta}_k+\overline{\mu}_k.
\end{split}
\end{equation}
\end{lemma}

\begin{proof}Similarly to the proof of Lemma \ref{lemma:2:2}, according to assumptions (\ref{zhen8})-(\ref{fg22}),
for $\overline{T}_3$ small enough, we have
\begin{equation}\label{eq:2.44t}
\begin{split}
&|\|U^{(k+1)}-U^{(k)}\||_{0,\overline{T}_3}\\
\leq & \overline{\eta}_1 \left(|\|U^{(k)}-U^{(k-1)}\||_{0,\overline{T}_3}+\left
(\int_0^\infty \int_{S^{2}}\||I^{(k)}-I^{(k-1)}\||^2_{0,\overline{T}_3} \text{d}\Omega \text{d}v\right)^{\frac{1}{2}}\right)+\overline{\beta}_k,
\end{split}
\end{equation}
where $\overline{\eta}_1 < \frac{1}{2}$ and $\sum_{k} |\overline{\beta}_k|< +\infty $.

To bound $I^{(k+1)}-I^{(k)}$, we use equation (\ref{eq:2.13}) to show that
\begin{equation}\label{eq:2.45tt}
\begin{split}
&\frac{1}{c} \partial_t (I^{(k+1)}-I^{(k)}) +\Omega\cdot\nabla (I^{(k+1)}-I^{(k)})\\
=&S^{(k+1)}-S^{(k)}-\sigma^{(k)}_a(I^{(k+1)}-I^{(k)})
-I^{(k)}(\sigma^{(k)}_a-\sigma^{(k-1)}_a)+\overline{J}_8+\overline{J}_9,
\end{split}
\end{equation}
where
\begin{equation*}
\begin{split}
\overline{J}_8=&\int_0^\infty \int_{S^{2}} \frac{v}{v'} \Big\{\sigma^{(k)}_s(I'^{(k)}-I'^{(k-1)})+I'^{(k-1)}(\sigma^{(k)}_s-\sigma^{(k-1)}_s)\Big\} \text{d}\Omega' \text{d}v',\\
\overline{J}_9=&\int_0^\infty \int_{S^{2}} \Big\{\sigma'^{(k)}_s(I^{(k+1)}-I^{(k)})+I^{(k)}(\sigma'^{(k)}_s-\sigma'^{(k-1)}_s)\Big\} \text{d}\Omega' \text{d}v'.
\end{split}
\end{equation*}
Similarly to the proof of Lemma \ref{lemma:2:2}, $\forall \tau \in [0,\overline{T}_*]$, we have
\begin{equation*}\begin{split}
& \int_0^\infty \int_{S^{2}}|\|(I^{(k+1)}-I^{(k)})\||^2_{0,\tau} \text{d}\Omega \text{d}v
\leq e^{C\tau}\Big(\int_0^\infty \int_{S^{2}}\|I^{(k+1)}_0-I^{(k)}_0\|^2 \text{d}\Omega \text{d}v \\
&\qquad\qquad\qquad\qquad\qquad+\tau\Big(|\|U^{(k)}-U^{(k-1)}\||^2_{0,\tau}
+\int_0^\infty \int_{S^{2}}|\|I^{(k)}-I^{(k-1)}\||^2_{0,\tau} \text{d}\Omega \text{d}v\Big)\Big).
\end{split}
\end{equation*}
Choosing $\overline{T}_4\in [0,\overline{T}_*]$ to be small enough, we have
\begin{equation}\label{eq:2.46tt}
\begin{split}
&\left(\int_0^\infty \int_{S^{2}}|\|I^{(k+1)}-I^{(k)}\||^2_{0,T_4} \text{d}\Omega \text{d}v\right)^{\frac{1}{2}}\\
\leq &\overline{\eta}_2\left(|\|U^{(k)}-U^{(k-1)}\||_{0,T_4}
+\left(\int_0^\infty \int_{S^{2}}|\|I^{(k)}-I^{(k-1)}\||^2_{0,T_4} \text{d}\Omega \text{d}v\right)^{\frac{1}{2}}\right)+\overline{\mu}_k,
\end{split}
\end{equation}
where $\overline{\eta}_2 < \frac{1}{2}$ and $\sum_{k} |\overline{\mu}_k|< +\infty $. 

Finally, taking $\overline{T}_{\ast\ast}=\min\{\overline{T}_3,\overline{T}_4\}$, Lemma \ref{lemma:2:2tt} is proved by adding (\ref{eq:2.44t}) and (\ref{eq:2.46tt}) together.
\end{proof}

Based on Lemmas \ref{lemma:2:1tt} and \ref{lemma:2:2tt}, we can prove Theorem \ref{th:1rr} analogously. We omit the details here.

It turns out that we have the following local existence of regular solutions to the original Cauchy problem (\ref{eq:1.2}) and (\ref{eq:2.211}) when $\sigma_s\neq 0$.
\begin{theorem}\label{th:1rre}
Let $s\geq 3$ be an integer and  (\ref{fg11})-(\ref{fg22}) hold.
If the initial data satisfy
\begin{equation*}\begin{split}
(I_0,\rho_0,u_0) \in \overline{\Psi}:&=\left\{(I,\rho,u )|\rho\geq0,\ (\rho^{\frac{\gamma-1}{2}},u)\in H^{s}, \ I(v,\Omega,x) \in L^2(\mathbb{R}^+\times S^{2};H^{s}(\mathbb{R}^3))\right\},
\end{split}
\end{equation*}
then there exists a time $T > 0$ such that the Cauchy problem (\ref{eq:1.2}) and (\ref{eq:2.211}) has a unique regular solution $(I,\rho,u)$.
\end{theorem}
The proof is the same as the corresponding theorem in Section $2$, here we omit it.


\begin{thebibliography}{99}

\bibitem{add1} C. Buet and  B. Despr$\acute{\text{e}}$s, Asymptotic analysis of fluid models for the coupling of radiation hydrodynamics, \textit{ J. Quant. Spectroscopy Rad. Transf.}  \textbf{85}  (2004) 385-418. \\[2pt]




\bibitem{add2} B. Ducomet, E., Feireisl and  $\check{\text{S}}$. Ne$\check{\text{c}}$asov$\acute{\text{a}}$, On a model in radiation hydrodynamics  , \textit{ Ann. Inst. H. Poincar$\acute{\text{e}}$. (C) Non Line. Anal.}  \textbf{6}  (2011) 797-812. \\[2pt]







%

\bibitem{BD} B. Ducomet and  $\check{\text{S}}$. Ne$\check{\text{c}}$asov$\acute{\text{a}}$, Global weak solutions to the 1-D compressible Navier-Stokes equations with radiation, \textit{ Commun. Math. Anal}  \textbf{8}  (2010) 23-65. \\[2pt]

\bibitem{BS} B. Ducomet and $\check{\text{S}}$. Ne$\check{\text{c}}$asov$\acute{\text{a}}$, Large time behavior of the motion of a viscous heat-conducting one-dimensional gas coupled to radiation,\textit{ Annali di Matematica Pura ed Applicata } \textbf{8} (2010) 219-260. \\[2pt]


\bibitem{pjd} P. Jiang and  D. Wang, Formation of singularities of solutions to the three-dimensional Euler-Boltzmann equations in radiation hydrodynamics, \textit{Nonlinearity} \textbf{23} (2010) 809-821.\\[2pt]

\bibitem{sjx} S. Jiang and  X. Zhong, Local existence and fiinte-time blow-up in multidimensional radiation hydrodynamics, \textit{J.\ Math.\ Fluid.\ Mech.} \textbf{9} (2007 ) 543-64.\\[2pt]

\bibitem{tk} T. Kato, The canchy problem for quasilinear symmetric hyperbolic systems, \textit{Arch.Ration.Mech.Anal.}  \textbf{ 58} (1945) 181-205. \\[2pt]
\bibitem{kr} R. Kippenhahn and  A. Weigert, \textit{Stellar structure and Evolution}. Springer: Berlin, Heideberg (1994). \\[2pt]



\bibitem{lax} P. Lax, Development of singularities of solutions of nonlinear hyperbolic
   partial differential equations,  \textit{J. Math. Phys.},
\textbf{5} (1964) 611--613.\\[2pt]

  \bibitem{lz} Y. Li and S. Zhu, Formation of singularities in solutions to the compressible radiation hydrodynamics equations  with vacuum, \textit{J. Differential Equations} \textbf{256} (2014) 3943-3980.\\[2pt]
\bibitem{zl} Y. Li and S. Zhu, Existence results for   the compressible radiation hydrodynamics equations with vacuum, 2013, submitted.  \\[2pt]






\bibitem{tpy} T. Liu and  T. Yang, Compressible Euler equations with vacuum, \textit{J.\ Differential Equations} \textbf{140} (1997) 223-237.\\[2pt]

\bibitem{tpl2} T. Liu, Compressible flow with damping and vacuum, \textit{Japan J. Indust.Appl.Math.} 13 (1996) 25-32.\\[2pt]   

\bibitem{amj} A. Majda, \textit{Compressible fluid flow and systems of conservation laws in several space variables}, Applied Mathematical Science \textbf{53}, Spinger-Verlag: New York, Berlin Heidelberg, 1986.\\[2pt]


\bibitem{tms3} T. Makino, Blowing up solutions of the Euler-Possion equation for the evolution of gaseous stars,\textit{ Trans.\ Theo.\ Statist.\ Phys.} \textbf{21} (1992) 615-624.\\[2pt]
\bibitem{tms1} T. Makino, S. Ukai, and S. Kawashima, Sur la solution $\grave{\text{a}}$ support compact de equations d'Euler compressible, \textit{Japan. J. Appl. Math} \textbf{33} (1986) 249-257.\\[2pt]
\bibitem{vn} J. Von Neumann, Discussion on the existence and uniqueness or multiplicity of solutions of the aerodynamical equations, \textit{Collected works of J. Von Neumann }, 1949.\\[2pt]

\bibitem{gp} G. Pomrancing, \textit{The Equations of Radiation Hydrodynamics}, Oxford: Pergamon, 1973.\\[2pt]


\bibitem{tbd} T. Sideris, T.Becca, and D. Wang, Long time behavior of solutions to the 3D compressible Euler equations with Damping, \textit{Commun.\ Part.\ Differ.\ Equations} \textbf{28:3-4} (2003) 795-816. \\[2pt]
\bibitem{ts2} T. Sideris, Formation of singulirities in three-dimensional compressible fluids, \textit{Comm.Math.Phys.} \textbf{101} (1985) 475-487.\\[2pt]


\bibitem{cy} C. Xu and  T. Yang, Local existence with physical vacuum boundary condition to Euler equations with damping,\textit{ J.\ Differertial Equations} \textbf{210} (2005) 217-231.\\[2pt]


%







%
















%


\bibitem{xy} Z. Xin and W. Yan, On blow-up of classical solutions to the compressible Navier-Stokes Equations, \textit{Commun. Math. Phys.}  \textbf{321} (2013) 529-541. 

\end{thebibliography}
\end{document}